%% file: full.tex
\documentclass[a4paper,10pt]{article}
\usepackage[utf8]{inputenc}
\usepackage{enumerate}
\usepackage[OT4]{fontenc}
\usepackage{graphicx}
\usepackage{amssymb}
\usepackage{amsmath}
\usepackage{amsfonts}
\usepackage{amsthm}
\usepackage{algorithm}
\usepackage[noend]{algpseudocode}
\usepackage{algorithmicx}
\usepackage{tikz}
\usepackage[all,cmtip]{xy}
\usepackage{caption}
\usepackage{tikzscale}
\usepackage{authblk}
\usepackage{fullpage}
\usepackage{enumitem}
\usepackage[hidelinks,hypertexnames=false]{hyperref}

\usetikzlibrary{backgrounds}
\usetikzlibrary{patterns}
\usetikzlibrary{decorations.pathmorphing}
\usetikzlibrary{decorations.pathreplacing}

\newlength{\xywd}
\newcommand{\xyrightarrow}[2][]{\sbox{0}{$\scriptstyle#1$}\xywd=\wd0
  \sbox{0}{$\scriptstyle#2$}\ifdim\wd0>\xywd \xywd=\wd0 \fi
  \xymatrix@C\dimexpr\xywd+1em\relax{{}\ar[r]^{#2}_{#1}&{}}}

\newtheorem{theorem}{Theorem}[section]
\newtheorem{lemma}[theorem]{Lemma}

\newtheorem{fact}[theorem]{Fact}
\newtheorem{definition}[theorem]{Definition}
\newtheorem{corollary}[theorem]{Corollary}

\theoremstyle{remark}
\newtheorem{remark}[theorem]{Remark}

\newcommand{\dual}[1]{\ensuremath{#1^{*}}}

\newcommand{\contr}[2]{\ensuremath{#1 / #2}}
\newcommand{\remove}[2]{\ensuremath{#1 - #2}}

\newcommand{\actcols}[1]{\ensuremath{\mathrm{act}(#1)}}
\newcommand{\aparts}{\ensuremath{\mathcal{A}}}
\newcommand{\prowmap}{\ensuremath{\mathrm{r}\pi}}
\newcommand{\pcolmap}{\ensuremath{\mathrm{c}\pi}}

\newcommand{\cset}{\ensuremath{\mathrm{can}}}

\newcommand{\rtree}{\ensuremath{\mathcal{T}}}
\newcommand{\bnd}[1]{\ensuremath{\partial #1}}

\newcommand{\intc}{\ensuremath{\mathrm{In}}}
\newcommand{\extc}{\ensuremath{\mathrm{Ex}}}
\newcommand{\gtc}{\ensuremath{\mathcal{R}}}
\newcommand{\leaftc}{\ensuremath{\mathrm{In}^*}}
\newcommand{\child}{\ensuremath{\mathrm{child}}}
\newcommand{\swon}[1]{\ensuremath{\overline{#1}}}
\newcommand{\curv}{\ensuremath{\mathcal{C}}}

\newcommand{\dirpath}[3]{\ensuremath{#1 \xyrightarrow{\scriptscriptstyle{#3}} #2}}

\begin{document}

\begin{titlepage}
\date{}
  \title{Decremental Single-Source Reachability in Planar Digraphs}
  \author[1]{Giuseppe F. Italiano\thanks{Partly supported by the Italian Ministry of Education,
  University and Research under Project AMANDA (Algorithmics for MAssive and Networked DAta).}}
  \author[2]{Adam Karczmarz\thanks{Supported by the Polish National Science Center grant number 2014/13/B/ST6/01811.}}
  \author[3]{\\Jakub Łącki\thanks{This work was done when the author was partly supported by the EU FET project MULTIPLEX no. 317532, the Google Focused Award on "Algorithms for Large-scale Data Analysis" and Polish National Science Center grant number 2014/13/B/ST6/01811. Part of this work was done while Jakub Łącki was visiting the Simons Institute for the Theory of Computing.}}
  \author[2]{Piotr Sankowski\thanks{Supported by Polish National Science Center grant number 2014/13/B/ST6/01811. Part of this work was done while Piotr Sankowski was visiting the Simons Institute for the Theory of Computing.}\bigskip}

\affil[1]{University of Rome Tor Vergata, Italy}
\affil[ ]{\texttt{giuseppe.italiano@uniroma2.it}\medskip}

\affil[2]{University of Warsaw, Poland}
\affil[ ]{\texttt{\{a.karczmarz|sank\}@mimuw.edu.pl} \medskip}

\affil[3]{Google Research, New York, USA}
\affil[ ]{\texttt{jlacki@google.com}}

\maketitle

  \input{abstract}

\begingroup
\renewcommand{\thefootnote}{}
\footnotetext[0]{\textit{Revision note (October 2026).}
This version corrects minor errors throughout the paper and fills gaps in the construction and analysis of a simple
recursive decomposition, providing additional details where needed.
The corrections are based on Karczmarz's 2018 PhD thesis~\cite{Karczmarz18Thesis}, which also covers
the results of this paper. The revision was assisted by GPT-6 Astra. All changes have been verified by AK.}
\endgroup

\end{titlepage}

\section{Introduction}\label{sec:introduction}
\input{introduction}

\input{preliminaries}

\input{switch-on}

\input{matrix}

\input{monge-trans}

\input{dag-trans}

\section{Extensions of the Switch-On Reachability Data Structure}

\subsection{Computing Maximal 2-Edge Connected Subgraphs}
\input{2conn}

\input{incremental}

\input{nonsimple}

\bibliographystyle{plain}
\bibliography{references}

\end{document}

%% file: abstract.tex
\begin{abstract}
In this paper we show a new algorithm for the decremental single-source reachability problem in directed planar graphs.
It processes any sequence of edge deletions in $O(n\log^2{n}\log\log{n})$ total time and explicitly maintains the set of vertices reachable from a fixed source vertex.
Hence, if all edges are eventually deleted, the amortized time of processing each edge deletion is only $O(\log^2 n \log \log n)$, which improves upon a previously known $O(\sqrt{n}\,)$ solution.
We also show an algorithm for decremental maintenance of strongly connected components in directed planar graphs with the same total update time.
These results constitute the first almost optimal (up to polylogarithmic factors) algorithms for both problems.

To the best of our knowledge, these are the first dynamic algorithms
  with polylogarithmic update times on general directed planar graphs
  for non-trivial reachability-type problems, for which
  only polynomial bounds are known in general graphs.

\end{abstract}

%% file: introduction.tex
The design of dynamic graph algorithms and data structures is one of
the classical areas in theoretical computer science.
Dynamic graph algorithms are designed to answer queries about a given property
while the underlying graph is subject to updates, such as inserting or deleting a vertex or an edge.
A dynamic algorithm is said to be \emph{incremental} if it handles only insertions,
\emph{decremental} if it handles only deletions,
and \emph{fully dynamic} if it handles both insertions and deletions.
Typically, one is interested in very small query times (either constant or polylogarithmic),
while minimizing the update times, whereas
the ultimate goal is to have both query and update times either constant or polylogarithmic.
This quest for obtaining polylogarithmic time algorithms has been so far successful only in few cases.
Indeed, efficient dynamic algorithms with \emph{polylogarithmic} time per update are known
only for few problems, and are mostly limited to undirected graphs, such as dynamic connectivity,
2-connectivity,  minimum spanning tree and maximal matchings
(see, e.g., \cite{bgs2015,HeTh97,HoLiTh01,Holm15,KaKiMo13,S16,Thorup00, Nilsen13}).
On the other hand, dynamic problems on directed graphs are notoriously harder.\footnote{
The only exception being incremental single-source reachability that has
trivial constant amortized update and query algorithm.}
For example, the fastest algorithms for basic dynamic problems like reachability
and transitive closure, have only \emph{polynomial} times per update
(see, e.g., \cite{DI08,K99,RZ08,S04}).
Similarly, polynomial algorithms are only known for dynamic shortest paths
that do not seem to be easier on undirected graphs than on directed ones~\cite{DI04,K99,T05}.

In this paper we consider the \emph{decremental single-source reachability} problem,
in which we are given a directed graph $G$ and a source node $s$ and the goal is to
maintain the set of vertices that are reachable from vertex $s$, subject to edge deletions.
Differently from undirected graphs, where polylog amortized bounds per update have been known
for more than one decade even in the fully dynamic setting (see, e.g., \cite{Thorup00}),
for directed graphs there has been very limited progress on this problem.
In fact, up until few years ago, the best known algorithm for decremental single-source reachability
had $O(1)$ query time and $O(mn)$ total update time
(i.e., $O(n)$ amortized time per update if all edges are deleted).
This bound was simply achieved with Even-Shiloach trees~\cite{ES81},
and it stood for over 30 years.
In a recent breakthrough, Henzinger et al.~\cite{HenzingerKN14} presented
a randomized decremental single-source reachability algorithm with
total update time $O(m n^{0.984 + o(1)})$, which
they later improved to $O(m n^{0.9 + o(1)})$ \cite{HenzingerKN15}.
Very recently, Chechik et al.~\cite{CHILP16} improved the total update time to
$\widetilde{O}(m \sqrt{n}\,)$.

A closely related problem to decremental single-source reachability
is \emph{decremental strong connectivity}, where we wish to answer queries of the form:
``Given two vertices $x$ and $y$, do $x$ and $y$ belong to the same strongly connected component?'', subject to edge deletions.
This problem is known to be almost equivalent to decremental single-source reachability
(see e.g.,~\cite{CHILP16,HenzingerKN14,RZ08}), and the randomized algorithm by
Chechik et al.~\cite{CHILP16} can solve also decremental strong connectivity in constant time per query and
$\widetilde{O}(\sqrt{n}\,)$ amortized time per update, over any sequence of $\Omega(m)$ deletions.
Again, the undirected version of this problem can be solved much faster, i.e.,
in polylog amortized time per update, even in the fully dynamic setting~\cite{Thorup00}.
Motivated by the limited progress on some dynamic graph problems,
in their seminal work Abboud and Vassilevska-Williams~\cite{AW14} proved conditional polynomial lower bounds,
based on popular conjectures, on several dynamic problems, including dynamic shortest paths,
dynamic single-source reachability and dynamic strong connectivity.

Similar situation holds even for planar graphs where dynamic problems have been studied extensively,
see e.g.~\cite{Abraham2012, Diks2007,  Eppstein96, Eppstein92,  Giammarresi96, Gustedt98, INSW11, scc-decomposition, steiner-tree, Lacki2011, decremental-connectivity,  Sub-ESA-93}.
Despite this effort, the best known algorithms for some basic problems on planar graphs,
such as dynamic shortest paths and decremental single-source reachability, 
still have polynomial update time bounds.
For instance, for dynamic shortest paths on planar graphs the best known bound per operation
is $\widetilde{O}(n^{2/3})$ amortized~\cite{FR06,GK16,INSW11,KMNS12,K05}.
Very recently, Abboud and Dahlgaard~\cite{AD16} proved polynomial update time lower bounds
for dynamic shortest paths on planar graphs, again based on popular conjectures.
In particular, they showed that obtaining $O(n^{1/2-\epsilon})$ bounds, for any $\epsilon>0$,
for dynamic shortest paths on planar graphs would yield a breakthrough for the all-pairs
shortest paths problem in general graphs.
Quite surprisingly, this lower bound almost matches
the best running times for two related problems: dynamic reachability~\cite{Diks2007}
and dynamic approximate shortest paths~\cite{Abraham2012}.
Hence, it might seem that the final answer to these problems is $\widetilde{O}(\sqrt{n})$.

Moreover, no polynomially faster algorithms are known even for decremental single-source
reachability on general planar graphs.
The algorithm by Łącki~\cite{scc-decomposition} solves both decremental single-source
reachability and decremental strong connectivity in a total of
$O(n\sqrt{n}\,)$ time, under any sequence of edge deletions. Note that this bound
is only logarithmic factors away from the $\widetilde{O}(m \sqrt{n}\,)$ bound on general graphs~\cite{CHILP16}.
Only in the very restricted case of $st$-planar graphs, i.e., planar acyclic
digraphs with exactly one source and exactly one sink, it was known for more than two decades
how to solve the dynamic reachability problem in $O(\log n)$ time per query and update~\cite{TP90,TT93},
provided that the edge insertions do not violate the embedding of the $st$-planar graph.

In this paper, we break through the natural $\widetilde{O}(\sqrt{n}\,)$ time barrier for
directed problems on planar graphs and present new decremental single-source reachability
and decremental strong connectivity algorithms for planar graphs with total update time of
$O(n\log^2{n}\log\log{n})$, i.e., in $O(\log^2{n}\log\log{n})$ amortized time per update,
over any sequence of $\Omega(n)$ deletions, and $O(1)$ time per query.
This result not only improves substantially the previously best known
$O(n\sqrt{n}\,)$ bound~\cite{scc-decomposition}, but it also
constitutes the first almost optimal (up to polylog factors) algorithm for those problems.
To the best of our knowledge, our result is the first nontrivial dynamic algorithm for
reachability problems on directed planar graphs with polylog update bound.
We hope that this result will pave a way for obtaining polylogarithmic
algorithms for dynamic directed problems on planar and possibly even general graphs.

\paragraph{Overview.}
Our improved algorithms for decremental
single-source reachability and decremental
\linebreak
strongly connected components are based on several new ideas and techniques.
First, we explore in a somewhat non-traditional way
the relation between primal and dual graph.
Indeed, in Section~\ref{sec:switch-on}, we show a formal reduction from decremental single-source reachability
to the \emph{switch-on reachability} problem that we introduce.
In the switch-on reachability problem we are given a directed graph, where all edges are
initially off and can be switched on in a dynamic fashion.
The goal is to maintain, for each edge $uw$, whether there is a path from
$w$ to $u$ that consists solely of edges that are on.

In Section~\ref{sec:matrix}, we analyze the structural properties of a reachability matrix of a set of vertices\footnote{
A reachability matrix of a set $S$ is a submatrix of the transitive closure matrix
consisting of rows and columns corresponding to vertices belonging to $S$.} that, roughly speaking,
lie on a constant number of simple faces, i.e., faces bounded with a single simple cycle.
We call such a matrix a face reachability matrix.
The matrix viewpoint allows us to obtain properties that are algorithmically useful and otherwise not easy to capture
using the previously used \emph{separating path} approach to
reachability in planar digraphs \cite{Diks2007, Sub-ESA-93, Thorup:2004}.
Those results are instrumental for 
designing a novel algorithm that can incrementally maintain the transitive closure
of a graph defined by the union of two face reachability matrices that undergo incremental updates.
This algorithm is described in Section~\ref{sec:monge-trans} and is perhaps one of the key technical achievements of this paper.

In order to solve the switch-on reachability problem, we combine the algorithm of Section~\ref{sec:monge-trans}
with a \emph{recursive decomposition} of a planar graph $G$,
which is a tree-like hierarchy of subgraphs of a graph $G$ (pieces) built by
recursively partitioning $G$ with small separators.
For each piece $H$, the size of the boundary $\bnd{H}$ of $H$ (i.e. the set of vertices of $H$ shared with other pieces that are neither ancestors nor descendants of $H$) is small.
Moreover, we require the boundary of $H$ to lie on a constant number of faces of $H$.
Such decompositions proved very useful in obtaining near-linear planar
graph algorithms, e.g., \cite{Borradaile:2015, INSW11, Lacki:2012, Lacki2011},
as well as dynamic planar graph algorithms \cite{Diks2007, FR06, Mozes:2012}.
For each piece $H$ of the decomposition, we dynamically maintain the face reachability matrix of
the set $\bnd{H}$ in both $H$ and its edge complement $G-H$.
The idea of maintaining the information about the complements of the pieces of the decomposition
has been used previously \cite{Borradaile:2015, Lacki:2012, Mozes:2012}.
However, to the best of our knowledge, in our paper this idea is used for the first time in a dynamic setting.
Moreover, we observe that in order to benefit from the Monge property of paths in the complement $G-H$,
the bounding walks of faces constituting the boundary of $H$ have to be simple cycles.
Unfortunately, none of the state-of-the-art recursive decomposition algorithms \cite{Borradaile:2015, Klein:13}
produces decompositions with such a property.
To overcome this issue, in Section~\ref{sec:preliminaries} we
propose an improved recursive decomposition that guarantees this property by possibly extending the input graph~$G$.
We believe that this issue has been overlooked in the previous works that
used the Monge property for complements of pieces \cite{Borradaile:2015, Lacki:2012, Mozes:2012}.
Finally, Section~\ref{sec:dag-trans} puts all the ingredients together
to obtain our $O(n\log^2{n}\log\log{n})$ time algorithm for the switch-on reachability problem.

Our result can be further extended to decrementally
maintain~the \emph{maximal (wrt. inclusion) 2-edge-connected subgraphs} of a planar digraph.
For $C\subseteq V$, an induced subgraph $G[C]$ is 2-edge-connected
iff it is strongly connected and contains no edges
whose removal would make $G[C]$ no longer strongly connected.
The 2-edge-connectivity in digraphs has been often studied in recent
years (see e.g., \cite{Chechik:2017, Georgiadis:2016, Henzinger:2015}).
In particular, the \emph{static} computation of the maximal 2-edge-connected subgraphs
proved to be a challenging task.
The best known bounds for general digraphs are $O(n^2)$ for dense graphs \cite{Henzinger:2015}
and $O(m^{3/2})$ for sparse graphs \cite{Chechik:2017}.
We show that for planar digraphs, the maximal 2-edge-connected subgraphs
can be computed in $O(n\log^2{n}\log\log{n})$ time.
Moreover, they can be maintained subject to edge deletions
within the same total time bound.

Another extension is an incremental transitive closure algorithm,
in which the embedding of the final graph is given upfront.
In other words, we are given a directed planar graph in which all edges are off,
each update operation switches some edge on and the goal is to answer reachability
queries with respect to the edges that are currently on.
We note that the same ``switch-on'' model was considered in~\cite{Gustedt98} for undirected reachability.
The reachability queries are answered in $\widetilde{O}(\sqrt{n})$ time, and any sequence of updates is
handled in $O(n\log^2{n}\log\log{n})$ time.
This result shows that progress can be made also in the case of
incremental planar graph transitive closure.
Moreover, it pinpoints the hardness of incremental planar graphs problems,
where few results have been obtained so far, as it is not known how to incrementally maintain small separators.

%% file: preliminaries.tex
\section{Preliminaries}\label{sec:preliminaries}
Let $G = (V, E)$ be a graph.
We let $V(G)$ and $E(G)$ denote the vertex and edge set of $G$, respectively.
Assume that $G$ is a \emph{directed} graph (digraph).
A set of vertices $S \subseteq V$ is \emph{strongly connected} if there is a path in $G$ between every two vertices of $S$.
A \emph{strongly connected component} (SCC) of $G$ is a
maximal (w.r.t. inclusion) strongly connected subset of vertices.

Throughout, 
we use the term planar digraph to denote a directed planar multigraph.
In particular, we allow the digraphs to have parallel edges and self-loops.
Formally, there might exist multiple edges $e_1,e_2,\ldots$, such that
each $e_i$ connects the same pair of vertices $u,v\in V$.
We use the notation $uv\in E$ to denote any of the edges $uv$,
whereas when we write $uv=e\in E$, we mean some specific edge $e$
going from $u$ to $v$.

Bounds stated in terms of the number $n$ of input vertices assume that the input
has $O(n)$ edges, as is the case for simple planar digraphs.
We retain parallel edges and self-loops.
Auxiliary graphs need not have a linear number of edges in their own number of vertices;
their size is measured by the total number of vertices and edges.
The manipulations we perform keep this total size $O(n)$.

Let $e \in E(G)$.
We denote by $\remove{G}{e}$ the graph obtained from $G$ by removing $e$ and, when $e$ is not a self-loop, by $\contr{G}{e}$ the graph obtained by contracting $e$.

Finally, if $G$ is a directed graph and $u, w \in V(G)$, we use $\dirpath{u}{w}{G}$ to denote a directed path from $u$ to $w$ in $G$.
The graph is sometimes omitted, if it is clear from the context.
\paragraph{Planar Graph Embeddings and Duality.}
We provide intuitive definitions below.
For formal definitions, see e.g.~\cite{diestel}. An embedding of a planar graph is a mapping of its vertices to distinct points
and of its edges to non-crossing curves in the plane.
We say that a planar graph $G$ is \emph{plane embedded} (or \emph{plane}, in short),
if some embedding of $G$ is assumed.
The face of a plane graph $G$ is a maximal open connected set of points that are not
in the image of any vertex or edge in the embedding of $G$.
There is exactly one \emph{unbounded} face.

The \emph{bounding walk} of a bounded (unbounded, resp.)
face $f$ of a connected plane graph is a closed walk consisting of edges bounding
$f$ in clockwise (counterclockwise, resp.) order. Here, we ignore the directions of edges.
The face is called \emph{simple} if and only if its bounding
walk is a simple cycle.

By the Jordan Curve Theorem, a Jordan curve $\curv$ partitions
$\mathbb{R}^2\setminus \curv$ into two connected
regions, a bounded one $B$ and an unbounded one $U$.
We say that a set of points $P$ is \emph{strictly inside} (\emph{strictly outside}, resp.) $\curv$,
if and only if $P\subseteq B$ ($P\subseteq U$, resp.).
$P$ is \emph{weakly inside} (\emph{weakly outside}, resp.) if and only if $P\subseteq B\cup\curv$
($P\subseteq U\cup\curv$, resp.).

Let $G=(V,E)$ be a plane embedded planar digraph.
The \emph{dual graph} of $G$, denoted by $\dual{G}$, is a plane directed graph, whose set of vertices is the set of faces of $G$.
Moreover, for each edge $uw=e\in E(G)$, $\dual{G}$ contains a directed edge from the face left of $e$ (looking from $u$ in the direction of $w$) to the face right of $e$.
We denote this edge in the dual as $e^*$.

We use the fact that removing an edge while preserving connectivity
in the primal graph corresponds to contracting its dual edge in the dual graph. Formally:

\begin{fact}\label{fact:remove-contr}
Let $G$ be a connected plane graph and $e\in E(G)$ such that $G-e$ is connected. Then
$\dual{(\remove{G}{e})} = \contr{\dual{G}}{\dual{e}}$.
\end{fact}

Let $G$ be a directed graph.
We say that $uw=e \in E$ is an \emph{inter-SCC} edge if $u$ and $w$ belong to distinct SCCs of $G$, and an \emph{intra-SCC} edge otherwise. Note that $uw$ is an inter-SCC edge iff
there is no path from $w$ to $u$ in $G$.
In our algorithm we use the following 
relation.

\begin{lemma}[folklore, see e.g.~\cite{Kao:93}]\label{lem:inter-dual}
Let $G$ be a plane embedded digraph and $e \in E(G)$.
Then, $e$ is an inter-SCC edge of $G$ iff $\dual{e}$ is an intra-SCC edge of $\dual{G}$.
\end{lemma}

\paragraph{Subgraphs, etc.}
Let $G=(V,E)$ be a plane (directed or undirected) graph.
A plane graph $G'=(V',E')$ is called a \emph{subgraph} of $G$
if $V'\subseteq V$ and $E'\subseteq E$.
$G'$ inherits the embedding of $G$.
For a subset $F\subseteq E$ we define an \emph{edge-induced subgraph} $G[F]$
to be a subgraph $(V_F,F)$ of $G$ such that $V_F$ is the set of endpoints of the edges~$F$.

For two graphs $G_1=(V_1,E_1), G_2=(V_2,E_2)$, we define
their \emph{union} $G_1\cup G_2$ as the graph $(V_1\cup V_2, E_1\cup E_2)$.
Analogously, we define the \emph{intersection} $G_1\cap G_2$ as
the graph $(V_1\cap V_2,E_1\cap E_2)$.

Let $H$ be a subgraph of $G$.
We define a \emph{hole} of $H$ as a face of $H$
that is not a face of $G$.
A \emph{simple hole} is a hole of $H$ that is a simple
face of $H$.

\paragraph{Planar Graph Decompositions.}\label{sec:decomp}
Let $G=(V,E)$ be a connected, undirected plane graph. Let $n=|V|$.
A \emph{recursive decomposition} of $G$ is a collection of \emph{connected}, edge-induced subgraphs
of $G$ organized in a binary tree $\rtree(G)$.
We write $H \in \rtree(G)$ to denote that $H$, which is a subgraph of $G$, is a node of $\rtree(G)$.
The \emph{level} of $H \in \rtree(G)$ is defined as the number of edges on the path from the root to $H$.
The tree $\rtree(G)$ has the following properties:
\begin{itemize}
\item The root of $\rtree(G)$ is the graph $G$ itself.
\item A non-leaf subgraph $H\in\rtree(G)$ has exactly two children $\child_1(H),\child_2(H)$
  such that $\child_1(H)\cup \child_2(H)=H$.
\item We define the set of \emph{boundary vertices} of a subgraph
$H\in\rtree(G)$ (denoted by $\bnd{H}$) inductively:
  \begin{itemize}
  \item If $H$ is the root of $\rtree(G)$, that is, $H = G$, then $\bnd{H}=\emptyset$.
  \item Otherwise, let $P$ and $S$ be the parent and the sibling of $H$ in $\rtree(G)$, resp.
    Then $\bnd{H}=V(H)\cap (V(S)\cup\bnd{P})$.
  \end{itemize}
\item Each $H\in\rtree(G)$ has $O(1)$ holes and all vertices of $\bnd{H}$ lie on the holes of $H$.
\item If $i$ is the level of $H$, then $|\bnd{H}|=O(\sqrt{n}/c^i)$, for some $c>1$.
\item $\sum_{H\in\rtree(G)} |\bnd{H}|^2=O(n\log{n})$.
\item There are $O(n)$ leaf subgraphs in $\rtree(G)$ and each leaf subgraph has $O(1)$ edges.
\item $\rtree(G)$ has $O(\log{n})$ levels.
\end{itemize}
\begin{remark}
The definition of a recursive decomposition naturally extends to directed plane graphs.
However, each $H\in\rtree(G)$ is connected in the undirected sense, i.e., weakly-connected.
\end{remark}
\begin{corollary}\label{cor:rtree-bounds}
  Let $\rtree(G)$ be a recursive decomposition of a plane graph $G=(V,E)$. Then
  each $e\in E$ is contained in some leaf subgraph of $\rtree(G)$.
  Moreover, $\sum_{H\in\rtree(G)} |E(H)|=O(n\log{n})$.
\end{corollary}
\begin{definition}\label{def:simple-decomp}
We call a recursive decomposition $\rtree(G)$ \emph{simple} if and only if for each $H\in\rtree(G)$:
\begin{enumerate}
\item 
 the holes of $H$ are all simple and vertex-disjoint,
\item 
 if $H$ has a sibling $S$ in $\rtree(G)$, then $E(H)\cap E(S)=\emptyset$.
\end{enumerate}
\end{definition}

The cycle-separator approach of~\cite{Klein:13} yields the recursive decompositions used below; Section~\ref{sec:nonsimple} gives the construction and the required bounds.
However, these decompositions do not guarantee that the holes are simple, which is essential when using ``external'' reachability matrices.
In Sections~\ref{sec:dag-trans}~and~\ref{sec:nonsimple}
we show how a planar graph $G$ can be augmented into a graph $G'$
so that $G'$ has a simple recursive decomposition and
all the interesting reachability properties of $G$ are preserved in $G'$.
Both the augmentation and the decomposition can be computed in $O(n\log{n})$ time.
\paragraph{Problem Definitions.}
This paper deals with two decremental graph problems.
In both problems the input is a digraph and the goal is to design a data structure that supports two operations.
The first operation deletes a given edge from $G$.
The second operation is a query operation and its meaning is problem-dependent.

In the \emph{decremental strongly connected components (decremental SCC)} problem the query operation is given two vertices $u, w \in V(G)$ as parameters.
The goal is to determine whether $u$ and $w$ are in the same SCC of $G$.

In the \emph{decremental single-source reachability (decremental SSR)} problem a source vertex $s \in V(G)$ is given in the input in addition to the graph $G$.
Each query operation has a vertex $v \in V(G)$~as parameter and asks to determine whether there exists a path $\dirpath{s}{v}{G}$.

The efficiency measure of such data structures is the total time needed to process all edge deletions.
Unless stated otherwise, the queries are answered in constant time.

%% file: switch-on.tex
\section{The Switch-On Reachability Problem}\label{sec:switch-on}
The core part of our algorithm solves the following problem, which we call \emph{switch-on reachability} problem.
The input is a digraph $G=(V,E)$.
Each edge of this digraph is either \emph{on} or \emph{off}.
Initially, each edge of $G$ is off and an update operation may turn an edge on.
Turning edges off is not allowed.
The goal is to maintain, for each edge $uw$ of $G$ (regardless of whether it is on or off), whether there is a directed path from $w$ to $u$ consisting solely of edges that are on.
Note that since turning edges off is not allowed, once such a path appears, it may never disappear.

For a graph $G$ whose edges undergo switch-ons, we denote by $\swon{G}$ the spanning subgraph consisting of the edges that are on.
We use this notation also for auxiliary graphs and subgraphs, including pieces of a recursive decomposition; for $H\subseteq G$, we have $\swon{H}=H\cap\swon{G}$.

We first show that the decremental SCCs problem in planar digraphs can be reduced to the switch-on reachability problem.
\begin{lemma}\label{lem:switch-inter}
Assume there exists an algorithm for the switch-on reachability problem in planar digraphs,
with time and space bounds measured in the total number of vertices and edges.
Then, there exists an algorithm that maintains the set of inter-SCC edges of a planar digraph under edge deletions with the same asymptotic running time and space usage.
\end{lemma}
\begin{proof}
Consider a digraph $G$ subject to edge deletions. Our goal is
to maintain the set of inter-SCC edges in $G$.
We process each initial weakly connected component separately and retain in $G$ any edge
whose deletion is requested after it has become inter-SCC, removing it only from the reported set.
Such edges cannot lie on a directed cycle, even after further deletions, so retaining them
preserves the SCCs and the classification of surviving edges.
Thus $G$ stays connected and only intra-SCC edges are deleted; in particular, no bridges are deleted.
By Lemma~\ref{lem:inter-dual}, it suffices to maintain the set of intra-SCC edges in $\dual{G}$.
Recall that when $G$ undergoes deletions, $\dual{G}$ is subject to edge contractions (see Fact~\ref{fact:remove-contr}). To reduce to switch-on reachability, we build a graph $H$ such that (i) edge contractions in $\dual{G}$ can be simulated by switch-on operations on $H$, and (ii)
$H$ preserves the reachability information of $\dual{G}$, that is, each current dual vertex is represented by a strongly connected set of vertices of $\swon{H}$ and a path between vertices of $\dual{G}$ exists if and only if a path between the corresponding strongly connected sets exists in $\swon{H}$.

We do this as follows.
Initially, let $H=(V(\dual{G}),E(\dual{G})\cup E^R(\dual{G}))$ where $E^R(\dual{G})$ is the set of \emph{reverse edges}, containing a unique edge ${e^{*R}}=wu$ for each edge $uw=\dual{e}\in E(\dual{G})$.
To initialize $H$, we start from all edges
switched off, and then switch on all the edges in $E(\dual{G})$. Note that, after this preprocessing, $H$ trivially preserves the reachability information of $\dual{G}$.

When an edge $e$ is deleted from $G$ and consequently its dual edge $\dual{e}$ is contracted in $\dual{G}$, we update $H$ by switching on the reverse edge ${\dual{e}}^R$.
As a result, $\dual{G}$ contains a single vertex created by the contraction, while the two corresponding strongly connected sets in $\swon{H}$ merge into one, joined in both directions by the original dual edge and its reverse.
Thus, 
throughout the sequence of updates (i.e., contractions in $\dual{G}$ and switch-on operations in $H$), $H$ keeps preserving the reachability information of $\dual{G}$.

Recall that an edge $e=uw$ of $\dual{G}$ is an intra-SCC edge iff there exists a path from $w$ to $u$ in $\dual{G}$.
This in turn happens iff the original head of $e$ can reach its original tail in $\swon{H}$.
In order to test this condition, we can run the algorithm for the switch-on reachability problem on $H$.
The algorithm maintains whether the endpoints of every edge are connected with a directed path.
This allows us to maintain the set of intra-SCC edges in $\dual{G}$, which in turn gives the desired set of inter-SCC edges in $G$.

We next analyze the running time. The total size $|V(H)|+|E(H)|$ is $O(n)$ for an input with $n$ vertices and $O(n)$ edges.
Building $\dual{G}$ takes linear time. Initializing $H$ takes linear time, plus the time required to switch on all the edges in $E(\dual{G})$.
Moreover, each edge deletion in $G$ can be translated to a corresponding operation of the algorithm for switch on-reachability in constant time, as we can precompute pointers between the corresponding edges of $G$, $\dual{G}$ and $H$.
Indeed, the edges of $\dual{G}$ map to the edges of $H$ in a natural way, as $\dual{G}$ is a minor of $H$ (i.e., $\dual{G}$ can be obtained from $H$ by deleting and contracting edges).
Thus, the running time and the space usage are dominated by the algorithm for switch-on reachability.
\end{proof}

\begin{lemma}\label{lem:inter-scc}
Assume there exists an algorithm that maintains the set of inter-SCC edges of a planar digraph under edge deletions.
Then, there exists an algorithm for the decremental SCC problem in planar digraphs with the same asymptotic total running time and space usage.
The query time of the decremental SCC algorithm is constant.
\end{lemma}
\begin{proof}
Let $G$ be a planar digraph subject to edge deletions.
Define $H$ to be the undirected graph obtained from $G$ by first removing inter-SCC edges and then ignoring edge directions.
Observe that the connected components of $H$ are exactly the same as the SCCs~of~$G$.
Thus, to determine if two vertices are strongly connected in $G$ it suffices to check if they are in the same connected component of $H$.

Observe that if $uv$ is an inter-SCC edge of $G$, then $uv$ is an inter-SCC edge in any subgraph of $G$ that contains it.
Thus, once an edge becomes an inter-SCC edge in $G$, it remains an inter-SCC edge until it is deleted.

The algorithm maintaining inter-SCC edges of $G$ allows us to maintain $H$: whenever an edge $e$ of $H$ is deleted from $G$ or becomes an inter-SCC edge in $G$, it is deleted from $H$.
Thus, edges may only be deleted from $H$.
As a result, we can use the decremental connectivity algorithm of Łącki and Sankowski~\cite{decremental-connectivity} to maintain the connected components of $H$.
This allows us to answer connectivity queries in $H$, and, consequently, strong connectivity queries in $G$.
The total update time of the decremental connectivity algorithm is linear and each query is answered in constant time.
Thus, the update and query times and the space usage of the resulting algorithm are asymptotically the same as in the algorithm maintaining inter-SCC edges.
\end{proof}

\begin{remark}\label{rem:ext-scc}
The paper of Łącki and Sankowski~\cite{decremental-connectivity} also describes an algorithm that explicitly maintains the set of vertices in each connected component and runs in $O(n \log n)$ time.
By using their algorithm, we can obtain a decremental SCC algorithm that not only supports queries in constant time, but also explicitly maintains the set of vertices in each SCC.
\end{remark}

Moreover, a decremental SCC algorithm implies an algorithm for decremental single-source reachability.
In general graphs, this follows from a very simple reduction (see e.g., \cite{CHILP16}).
However, this reduction does not preserve planarity, so we need to provide a slightly more complicated one.
\begin{lemma}\label{lem:scc-ssr}
Assume there exists an algorithm for decremental SCC problem in planar digraphs that explicitly maintains the set of vertices in each SCC and runs in $\Omega(n \log n)$ time.
Then, there exists an algorithm for decremental single-source reachability problem in planar digraphs with the same asymptotic total update time, query time and space usage.
\end{lemma}
\begin{proof}
A \emph{condensation} of a directed graph $G$, denoted by $cond(G)$, is a graph obtained from $G$ by contracting all its SCCs.
The vertices of $cond(G)$ are sets of vertices of $G$ contained in the corresponding SCC.
Note that $cond(G)$ is a directed acyclic graph (DAG).

Our goal is to maintain the set of vertices of $cond(G)$, which are reachable from the vertex of $cond(G)$ containing the source vertex $s$.
The high-level idea is that in order to maintain SCCs of $G$ we use the decremental SCC algorithm, whereas to maintain the set of reachable vertices in $cond(G)$ we use the fact that it is a DAG, which makes the problem much easier.

We first describe a simple dynamic single-source reachability algorithm for DAGs.
This algorithm maintains the set of vertices reachable from a fixed source and supports two types of updates.
The first update removes a single edge, whereas the second one replaces a vertex $v$ with an acyclic subgraph $H$.
This happens in three steps.
First, some number of new vertices are added to the maintained graph $G$ (vertex $v$ is not deleted).
Second, some edges of $G$, whose endpoint is $v$ may change this endpoint to one of the newly added vertices.
Third, new edges can be added, but their endpoints can only be the newly added vertices and $v$.
Also, adding these edges may not introduce cycles.

Note that both operations are in a sense decremental, as once a vertex becomes unreachable from the source, it never becomes reachable again.
The set of vertices reachable from the source can be easily maintained by iteratively applying the following principle: if a vertex distinct from the source has no incoming edges, it is not reachable from the source and thus can be deleted from $G$.
The resulting algorithm runs in time which is linear in the size of the original graph, the number of vertices and edges added, and the number of edge-endpoint changes.
See~\cite{scc-decomposition} for details.

It remains to describe how to maintain $cond(G)$ and funnel the updates to $cond(G)$ to the dynamic single-source reachability algorithm for DAGs.
We run the decremental SCC algorithm that maintains the SCCs explicitly (see Remark~\ref{rem:ext-scc}).
Whenever an inter-SCC edge of $G$ is deleted, it has a corresponding edge in $cond(G)$ and to update $cond(G)$ it suffices to remove this corresponding edge.
Otherwise, if an intra-SCC edge is removed from an SCC $C$, the SCC may decompose into SCCs $C_1, \ldots, C_k$.
Wlog. assume that $C_1$ is the largest one of these SCCs.
Thus, to update $cond(G)$ we add one new vertex for each of $C_2, C_3, \ldots, C_k$.
We do not need to add a vertex corresponding to $C_1$, as we update the vertex corresponding to $C$ (reuse it), so that it represents $C_1$.
If $s\in C$, we designate the component containing $s$ as the source before applying the pruning rule.
Some edges that were incident to $C$ need to be updated, as after the edge deletion their endpoint is one of $C_2, \ldots, C_k$.
In order to do that, we iterate through all edges of $G$ incident to vertices contained in $C_2, \ldots, C_k$.
Similarly, by iterating through all these edges we may add all new inter-SCC edges that appear as a result of the edge deletion.

It follows that the total running time of the algorithm is dominated by the initial graph size and the total time of iterating through edges in the process of handling an edge deletion.
An edge $uw$ is considered only when the size of the SCC containing either $u$ or $w$ halves.
Thus, we spend $O(\log n)$ time for each edge, which gives $O(n \log n)$ total time.
\end{proof}

%% file: matrix.tex
\section{Structural Properties of Reachability in Plane Digraphs}\label{sec:matrix}
In this section we present the structural properties of reachability in planar digraphs that we later exploit in our algorithms.
We show how to efficiently represent reachability information between a set of vertices that, roughly speaking, lie on a constant number of faces.
In fact, our analysis is more general and extends to sets of vertices that lie on a constant number of \emph{separator curves}, which we define below.
In the following definition, we consider undirected graphs or directed graphs where edge directions are ignored.
\begin{definition}\label{def:separator-curve}
  Let $G$ be a plane embedded graph. A Jordan curve $\curv$ is a \emph{separator curve} of $G$
if and only if one of the following holds:
\begin{enumerate}
\item 
each connected component of $G$ lies either weakly inside $\curv$ or strictly outside $\curv$;
\item
each connected component of $G$ lies either weakly outside $\curv$ or strictly inside $\curv$. 
\end{enumerate}
  Moreover, for each $e\in E(G)$, the interior of (the embedding of) $e$ is either a contiguous
  fragment of $\curv$ or is disjoint with $\curv$.
\end{definition}

We will sometimes abuse the notation and identify the separator curve with the set of vertices lying on it.

\begin{fact}\label{fac:face-curve}
Let $f$ be a cycle bounding some simple face of a plane embedded graph $G$.
Then, the closed curve defined by the embedding of $f$ is a separator curve.
\end{fact}

Consider a plane digraph $G$ and let $U=U_1\cup\ldots\cup U_\ell$ be a set of vertices of $G$ that lie on $\ell=O(1)$ \emph{pairwise disjoint} separator curves $\curv_1,\ldots,\curv_\ell$ of $G$.
We have $U_i=U\cap \curv_i$, which implies $U_i\cap U_j=\emptyset$ for $i\neq j$.
We define a total order $\prec$ on the elements of $U$, which satisfies the following property.
Consider the sequence $S(U)$ of elements of $U$ sorted according to $\prec$.
Then, elements of each $U_i$ form a contiguous fragment of $S(U)$ and are sorted in clockwise order.
There could be multiple total orders that satisfy this property.
Namely, for each $U_i$ the smallest vertex with respect to $\prec$ can be chosen arbitrarily and the contiguous fragments corresponding to sets $U_1, \ldots, U_\ell$ may come in any order.
However, this is not relevant to our later analysis.
Thus, in the following, we assume that each set $U$ has an associated total order $\prec$.
For $X,Y\subseteq U$ we also write $X\prec Y$ if for each $x\in X$, $y\in Y$, we have $x\prec y$.
Note that $X\prec Y$ implies $X\cap Y=\emptyset$.
\begin{definition}\label{def:reach-matrix}
  Let $A$ be a binary
  matrix with both its rows and columns indexed with the vertices of $U=U_1\cup\ldots\cup U_\ell$,
  such that the sets $U_i$ lie on $\ell$ pairwise disjoint separator curves $\curv_1,\ldots,\curv_\ell$
  of $G$, respectively. The rows and columns of $A$ are ordered according to the order $\prec$ induced
  by the sets $U_1,\ldots,U_\ell$.
  We say that $A$ is a \emph{reachability matrix} for $U$ if and only if for each $u,v\in U$,
  $A_{u,v}=1$ if and only if there exists a path $\dirpath{u}{v}{G}$.
\end{definition}
In the following we fix the sequences $U_1,\ldots,U_\ell$ and $\curv_1,\ldots,\curv_\ell$
and the order $\prec$ satisfying the assumption of Definition~\ref{def:reach-matrix}.
We set $U=U_1\cup\ldots\cup U_\ell$.
\begin{definition}
Let $X,Y\subseteq U$ and let $A$ be the reachability matrix of $U$.
  A binary 
  matrix $A^{X,Y}$ with rows indexed with $X$ and columns indexed with $Y$
  is called a \emph{reachability submatrix} for $X,Y$ iff $A^{X,Y}_{x,y}=A_{x,y}$
  for all $x\in X$ and $y\in Y$.
  Also, denote by $A^X$ the matrix~$A^{X,X}$.
\end{definition}
\begin{definition}
Let $X,Y\subseteq U$.
The subset of $Y$ containing those $y$ such that for some $x\in X$,
$A^{X,Y}_{x,y}=1$ is called the set of \emph{active columns} of $A^{X,Y}$ and
  is denoted by $\actcols{A^{X,Y}}$.
\end{definition}

\begin{definition}For a set $\aparts=\{A^{S_1,T_1},\ldots,A^{S_k,T_k}\}$ of reachability submatrices
and a row $s\in V$, we define
a \emph{row projection} $\prowmap_s(\aparts)$ to be the subset $\{A^{S_j,T_j}\in\aparts : s\in S_j\}$.
Similarly, for $t\in V$, we define a \emph{column projection} $\pcolmap_t(\aparts)=\{A^{S_j,T_j}\in\aparts : t\in T_j\}$.
\end{definition}
The following lemma provides a decomposition of the reachability matrix $A=A^U$
used in the next sections in a black-box~manner.
\begin{lemma}\label{lem:monge_reach}
  Let $G$ be a plane embedded digraph and let $\curv_1,\ldots,\curv_\ell$ be pairwise disjoint separator curves of $G$, where $\ell = O(1)$.
  Let $U$ be a subset of $V(G)$ of size $m$ such that $U=U_1\cup\ldots\cup U_\ell$ and the vertices of $U_i$ lie on $\curv_i$.

  Then, the reachability matrix $A$ of $U$ in $G$ can be partitioned into a set $\aparts = \{A^{S_1,T_1}, \ldots, A^{S_k,T_k}\}$
  of reachability submatrices such that:
  \begin{enumerate}
  \item for each $s,t\in U$, there exists exactly one such $A^{S_j,T_j}\in\aparts$ that $s\in S_j$ and $t\in T_j$,
  \item for each $A^{S_j,T_j}\in\aparts$ and $T'=\actcols{A^{S_j,T_j}}$, the ones in each
    row of $A^{S_j,T'}$ form $O(1)$ blocks,
  \item for any $s\in U$ and $t\in U$, the sets
    $\prowmap_s(\aparts)$ and $\pcolmap_t(\aparts)$ have size $O(\log{m})$.
  \end{enumerate}
  The sets $S_j$ and $T_j$
  that define the partition do not depend on the entries of $A$.
  The partition $\aparts$ can be computed in $O(m^2)$ time.
\end{lemma}
The remaining part of the section is devoted to proving Lemma~\ref{lem:monge_reach}.
\begin{definition}
The reachability submatrix $A^{S,T}$ is called \emph{bipartite},
  if $S$ and $T$ lie on a single separator curve (i.e., $S,T\subseteq U_i$ for $i\in \{1,\ldots,\ell\}$)
  and either $S\prec T$ or $T\prec S$.
  For convenience, we also regard $1\times 1$ reachability submatrices as bipartite.
\end{definition}
\begin{lemma}[Monge property]\label{lem:monge}
  Let $A^{S,T}$ be a bipartite reachability submatrix of $U$.
  Let $a,b\in S$ and $c,d\in T$ be such that $a\prec b$ and $c\prec d$.
  Suppose $A^{S,T}_{a,c}=1$ and $A^{S,T}_{b,d}=1$.
  Then $A^{S,T}_{a,d}=1$ and $A^{S,T}_{b,c}=1$ also hold.
\end{lemma}
\begin{proof}
We have $S,T\subseteq U_i$ for some $i$. Without loss of generality, assume that $S\prec T$ and that
all connected components of $G$ lie either weakly inside
of $\curv_i$ or strictly outside $\curv_i$.
Since $a,b,c,d$ lie on $\curv_i$, the connected components of $G$
containing any of these vertices are all weakly inside $\curv_i$.
  As $A^{S,T}_{a,c}=1$, there exists a simple path $P = \dirpath{a}{c}{G}$
contained entirely weakly inside $\curv_i$.
Similarly, there exists a path $Q = \dirpath{b}{d}{G}$ contained weakly inside $\curv_i$.
By planarity of $G$, we conclude that the paths $P$ and
$Q$ must have a common vertex $w$ (see Figure~\ref{fig:monge}).
Note that $w$ is reachable from both $a$ and $b$.
Analogously, both $c$ and $d$ are reachable from $w$.
Thus, there also exist paths $\dirpath{a}{d}{G}$ and $\dirpath{b}{c}{G}$.
\end{proof}

\begin{figure}[htb]
  \begin{minipage}{.5\columnwidth}
    \centering
    \includegraphics[height=5cm]{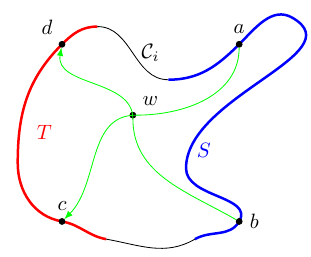}
    \caption{\label{fig:monge}}
  \end{minipage}\begin{minipage}{.5\columnwidth}
    \centering
    \includegraphics[height=5cm]{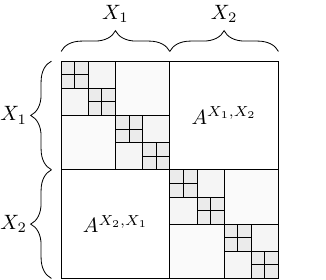}
    \caption{\label{fig:part}}
  \end{minipage}
\end{figure}

\begin{lemma}\label{lem:consec}
  Let $A^{S,T}$ be a bipartite reachability submatrix of $U$.
  Let $T'=\actcols{A^{S,T}}$.
  Denote by $t_1\prec\ldots\prec t_p$ the elements of $T'$.
  For each non-zero row $s$ of $A^{S,T'}$,
  the ones in that row form a single contiguous block, i.e., there exist
  $l_s$ and $r_s$ ($1\leq l_s\leq r_s\leq p$) such that $A^{S,T'}_{s,t}=1$
  holds if and only if $t\in \{t_{l_s},\ldots,t_{r_s}\}$.
\end{lemma}

\begin{proof}
Let $s\in S$ be some row of $A^{S,T'}$ containing at least one non-zero entry.
Suppose there exist $a,b,c\in T'$ such that $a\prec b\prec c$ and $A^{S,T'}_{s,a}=1$,
$A^{S,T'}_{s,b}=0$ and $A^{S,T'}_{s,c}=1$.
By the definition of $T'$, there exists some other row $s'\in S, s'\neq s$, such that
$A^{S,T'}_{s',b}=1$.
Wlog. assume $s\prec s'$.
Thus, we have both $A^{S,T'}_{s,a}=1$ and $A^{S,T'}_{s',b}=1$.
By Lemma~\ref{lem:monge} this implies $A^{S,T'}_{s,b}=1$, a contradiction.
\end{proof}

\begin{lemma}\label{lem:part}
 Let $U_i$ be a set of vertices lying on a single separator curve in $G$ and $m=|U_i|$.
  The reachability submatrix $A^{U_i}$ can be partitioned into a set $\aparts^{U_i} = \{A^{S_1,T_1}, \ldots, A^{S_k,T_k}\}$
  of bipartite reachability submatrices such that:
  \begin{enumerate}
  \item for each $s,t\in U_i$, there exists exactly one such
  $A^{S_j,T_j}\in\aparts^{U_i}$ that
    $s\in S_j$ and $t\in T_j$,
  \item for any $s\in U_i$ and $t\in U_i$, the sets
    $\prowmap_s(\aparts^{U_i})$ and $\pcolmap_t(\aparts^{U_i})$ have size $O(\log{m})$.
  \end{enumerate}
  The sets $S_j$ and $T_j$
  that define the partition do not depend on the entries of $A^{U_i}$.
 The partition $\aparts^{U_i}$ can be computed in $O(m^2)$ time.
\end{lemma}
\begin{proof}
We give a recursive procedure to construct the partition $\aparts^{X}$, for $X\subseteq U_i$.
If $|X|=1$, the procedure exits immediately with $\aparts^X=\{A^X\}$.
Otherwise, let $X=\{x_1,\ldots,x_k\}$, where $x_1\prec\ldots\prec x_k$ and $k\geq 2$.
  Let $q=\lfloor(k+1)/2\rfloor$.
  Set $X_1=\{x_1,\ldots,x_q\}$ and $X_2=\{x_{q+1},\ldots,x_k\}$.
  Note that $X_1\prec X_2$.
  Thus, both $A^{X_1,X_2}$ and $A^{X_2,X_1}$ are bipartite reachability submatrices.
  We add these matrices to the partition and recurse on the subsets
  $X_1$ and $X_2$ (see Figure~\ref{fig:part}).

  Note that for any $s,t\in U_i$, $s\neq t$, the last recursive call with both
  $s$ and $t$ in the input set $X$ places the entries $A^{U_i}_{s,t}$ and $A^{U_i}_{t,s}$
  in the bipartite reachability submatrices $A^{X_1,X_2}$ and $A^{X_2,X_1}$.
  Each diagonal entry is included in the singleton submatrix produced by the corresponding base case.

  Fix $x \in X$ and assume $|X| = k$.
  Let $f(k)$ be the number of bipartite reachability submatrices that are produced by the recursive algorithm and contain a row (column)
  corresponding to $x$.
  We have $f(k)\leq f(\lceil k/2\rceil)+1$ and thus it is easy to see that $f(k)=O(\log{k})$.
  Thus, we conclude that for any $s$, $|\prowmap_s(\aparts^{U_i})|=O(\log{m})$ and
  similarly $|\pcolmap_s(\aparts^{U_i})|=O(\log{m})$.

  The recursive procedure runs in time that is proportional to the total size of matrices that are produced.
  Recall that for each $s,t\in U_i$, there exists exactly one such $A^{S_j,T_j}\in\aparts^{U_i}$ that $s\in S_j$ and $t\in T_j$.
  It follows that the total running time is $O(m^2)$.
\end{proof}

An analogue of Lemma~\ref{lem:consec} can be shown for reachability
submatrices $A^{U_i,U_j}$, where $i\neq j$.
\begin{lemma}\label{lem:intra}
Let $A^{S,T}$ be a reachability submatrix for the sets $S=U_i,T=U_j$, where $i\neq j$.
  Let $T'=\actcols{A^{S,T}}$.
  For each row $s$ of $A^{S,T'}$,
  the ones in that row form $O(1)$ connected blocks.
\end{lemma}
\begin{proof}
  Let $s$ be some row of $A^{S,T}$ containing at least one non-zero entry.
  Equivalently, there exists a path $\dirpath{s}{t}{G}$ for some $t\in T$.
  Without loss of generality, assume that all connected components of $G$ lie either weakly inside
  $\curv_i$ or strictly outside $\curv_i$.
  As $s\in\curv_i$ and $t$ is in the same connected component of $G$ as $s$,
  $t$ lies weakly inside $\curv_i$.
  However, $t\in\curv_j$ and $\curv_i\cap\curv_j=\emptyset$, so entire $\curv_j$ lies
  strictly inside $\curv_i$.
  In particular, all vertices of $T$ lie strictly inside $\curv_i$.

  Let $P= \dirpath{s}{t}{G}$ be some shortest directed path from $s$ to a vertex of $T$.
  Clearly, $P$ is simple
  and $t$ is the only vertex of $V(P)\cap T$.
  Let $q$ be the last vertex of $P$ on $\curv_i$, and let $y$ be the first vertex
  of $P$ on $\curv_j$ after $q$. Denote by $Q$ the directed subpath $q=v_1\ldots v_k=y$ of $P$.
  We allow $q\notin S$ and $y\notin T$.
  By the definition of a separator curve, $Q$ is weakly inside $\curv_i$, weakly outside
  $\curv_j$ and $Q\cap\curv_i=\{q\}$, $Q\cap\curv_j=\{y\}$.

  Let $G'$ be obtained by ``cutting'' $G$ along the path $Q$ as follows.
  We split each vertex $v_i$ of $Q$ into two vertices $v_i'$ and $v_i''$,
  that inherit the edges of $v_i$ strictly emanating left or
  right of $Q$, respectively.
  The created vertices are connected with directed paths $Q'=v_1'\ldots v_k'$
  and $Q''=v_1''\ldots v_k''$ (see Figure~\ref{fig:inter-curve}).
  Let $U^*=(S\cup T)\setminus V(Q)\cup V(Q')\cup V(Q'')$.
  Note that all the vertices of $U^*$ lie on a single separator curve $\curv^*$ of $G'$.
  This curve imposes some clockwise order $\prec_s$ on $U^*$.

  First we prove that any $t\in T\setminus\{y\}$ is reachable from $q$ in $G$
  if and only if $t$ is reachable from either $q'$ or $q''$ in $G'$.
  The ``if'' part is trivial.
  Suppose there is a path $R=\dirpath{q}{t}{G}$ and let $z$ be the last
  vertex that $R$ and $Q$ have in common.
  Note that $z$ is not the last vertex of $R$, since $V(Q)\cap T\subseteq\{y\}$, and let the edge going out
  of $z$ emanate (wlog.) left of $Q$.
  The subpath $\dirpath{z}{t}{}$ of $R$ only shares its starting point with $Q$ and
  thus corresponds to a path $\dirpath{z'}{t}{G'}$.
  Therefore, there exists a path $\dirpath{q'}{z'}{G'} \overset{\scriptscriptstyle{G'}}{\longrightarrow} t$.

  Next we prove that for $s\neq q$, any $t\in T\setminus\{y\}$ is reachable from
  $s$ in $G$ if and only if $t$ is reachable from either $q'$ or $q''$ or $s$ in $G'$.
  To see the ``if'' part, note that $t$ is reachable from $q'$ or $q''$ in $G'$
  iff it is reachable from $q$ in $G$, whereas $q$ is reachable from $s$ in $G$.
  For the ``only if'' part, suppose there is a path $\dirpath{s}{t}{G}$.
  Either this path intersects with $Q$ and then $t$ is reachable from $q$ in $G$
  (which implies that $t$ is reachable from $q'$ or $q''$ in $G'$), or
  the path does not intersect with $Q$ and is thus preserved in~$G'$.

  Now let $S^*=S\setminus\{q\}\cup\{q',q''\}$ and $T^*=T\setminus\{y\}$.
  Observe that the sets $S^*,T^*$ constitute
  contiguous fragments of the cyclic order on $\curv^*$ restricted to $S^*\cup T^*$.
  Choose the cyclic shift of $\prec_s$ so that $S^*\prec_s T^*$.
  Consider the bipartite reachability matrix $A^{S^*,T^*}$ of the graph $G'$.
  We show that $\actcols{A^{S,T}}=\actcols{A^{S^*,T^*}}\cup(\{y\}\cap T)$.
  Since $q$ and $y$ are reachable from $s$, we have
  $\actcols{A^{S^*,T^*}}\cup(\{y\}\cap T)\subseteq \actcols{A^{S,T}}$.
  Let $x\in\actcols{A^{S,T}}\setminus\{y\}$.
  There exists a path $\dirpath{s'}{x}{G}$ for some $s'\in S$, where $x\in T\setminus\{y\}$.
  If this path does not intersect with $Q$, then it also exists in $G'$.
  Otherwise, $x$ is reachable from $q$ in $G$ and thus is reachable
  from either $q'$ or $q''$ in $G'$.

  By Lemma~\ref{lem:consec}, the ones in each row of $A^{S^*,{T^*}'}$, where
  ${T^*}'=\actcols{A^{S^*,T^*}}$, form at most one block in the order $\prec_s$.
  Now let $B$ be the matrix $A^{S^*,{T^*}'}$ but with columns ordered
  according to the order imposed on $T$.
  As the order $\prec_s$ on $T^*$ is, up to reversal, a cyclical shift of the original order on $T\setminus\{y\}$,
  we conclude that for each row of $B$, the ones in that row form $O(1)$ blocks
  of columns.
  To finish the proof, observe that the ones in the row $s$ of $A^{S,T'}$
  can be seen as the union of:
  \begin{itemize}
    \item a single one in the column $y$, if $y\in T$,
    \item the ones in the row $q'$ of $B$,
    \item the ones in the row $q''$ of $B$,
    \item if $s\neq q$, the ones in the row $s$ of $B$.
  \end{itemize}
  The rows of $B$ contribute $O(1)$ blocks in $T'\setminus\{y\}$.
  Inserting the column $y$ when $y\in T$ preserves the $O(1)$ bound.
\end{proof}

\begin{figure}
\centering
\includegraphics{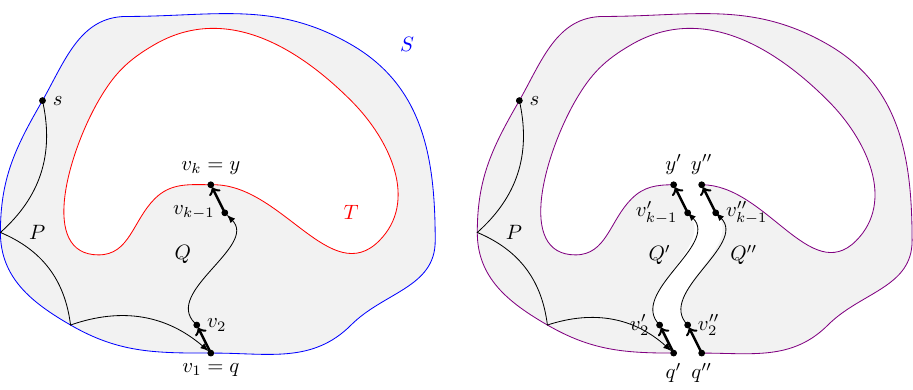}
\caption{}
\label{fig:inter-curve}
\end{figure}

\begin{proof}[Proof of Lemma~\ref{lem:monge_reach}]
  We first partition $A$ into $\ell^2$ reachability submatrices $A^{U_i,U_j}$,
  for all $i,j\in\{1,\ldots,\ell\}$.
  Each $A^{U_i,U_i}$ is then partitioned using Lemma~\ref{lem:part},
  whereas each $A^{U_i,U_j}$, for $i\neq j$, is included in $\aparts$
  without further partitioning.
  As each $u\in U$ belongs to exactly one $U_i$,
  $|\prowmap_u(\aparts)|=|\pcolmap_u(\aparts)|=O(\log{|U_i|})+\ell-1=O(\log{m})$.
  By Lemmas~\ref{lem:consec}~and~\ref{lem:intra}, for all $A^{S,T}\in\aparts$,
  the ones in each row of $A^{S,T'}$, where $T'=\actcols{A^{S,T}}$, form at most
  $O(1)$ connected blocks.
  The running time of this procedure is ${O(m^2 + \sum_i |U_i|^2)} = O(m^2)$.
\end{proof}

%% file: monge-trans.tex
\section{The Switch-On Monge Transitive Closure}\label{sec:monge-trans}
Let $G_1$ and $G_2$ be two plane embedded digraphs
and let $U_i\subseteq V(G_i)$
be a set of vertices lying on a constant number of pairwise disjoint separator
curves of $G_i$, for $i=1,2$.
Moreover, assume that 
$V(G_1)\cap V(G_2)\subseteq U_1\cap U_2$.

In this section we consider the following problem.
Assume that the edges of $G_1$ and $G_2$ undergo switch-on operations.
This way, the entries of the reachability matrix $A_i$ of $U_i$ in $\swon{G_i}$
may change from $0$ to $1$.
We wish to maintain the reachability matrix $A=A^{U_1\cup U_2}$ of
the graph $\swon{G_1}\cup \swon{G_2}$ subject to updates to the matrices
$A_1$ and $A_2$.
Note that our assumptions imply that $A$ depends only on
$A_1$ and~$A_2$.

Observe that as $\swon{G_i}$ is a subgraph of $G_i$, the separator curves
of $G_i$ are also the separator curves of $\swon{G_i}$.
The remaining part of this section is devoted to proving the following theorem.
\begin{theorem}\label{thm:trans}
Let $U_i$ be a set of vertices lying on a constant number of separator curves of $G_i$ (for $i=1,2$),
such that $V(G_1)\cap V(G_2)\subseteq U_1\cap U_2$.
Let $m=|U_1\cup U_2|$.
Assume that the edges of $G_i$ undergo switch-ons and denote by $\swon{G_i}$
the spanning subgraph of $G_i$ consisting of the edges that are on.
Let $A_i$ be the reachability matrix of $U_i$ in $\swon{G_i}$ and let
  $A$ be the reachability matrix of $U_1\cup U_2$ in $\swon{G_1}\cup \swon{G_2}$.
Suppose that after a switch-on we are given the list of entries in $A_1$ and $A_2$
that changed. Then we can initialize and maintain the matrix $A$ in $O(m^2\log{m}\log\log{m})$ total time.
\end{theorem}
\subsection{Queue-Based Incremental Transitive Closure Algorithm}

\newcommand{\unreachable}{\mathit{Unreachable}}
\newcommand{\cannotreach}{\mathit{CannotReach}}
We first show and analyze a simple queue-based algorithm for updating the transitive closure of a graph after a set of edges is added.
Its pseudocode is given as Algorithm~\ref{alg:inc-tc}.
The algorithm should be considered folklore, but for sake of completeness we describe it in a detailed way, as we need its efficient implementation.

The transitive closure algorithm is based on the following idea.
Whenever it determines that a vertex $b$ is reachable from a vertex $a$, it infers that every vertex reachable from $b$ by an edge is also reachable from $a$ and every vertex that has an edge to $a$ can also reach $b$.
Then, it propagates this information using a queue.
In the pseudocode, we use $Out(b)$ to denote the set of heads of out-edges of $b$ and $In(a)$ to denote the set of tails of in-edges of $a$.
Moreover, $\unreachable(a, A)$ is the set of vertices $x$ such that $A_{a,x} = 0$ and $\cannotreach(b, A)$ is the set of vertices $x$ such that $A_{x,b} = 0$.
\begin{algorithm}
\begin{algorithmic}[1]
\Require {A digraph $G$, such that $E^{+} \subseteq E(G)$ and the transitive closure $A$ of $\remove{G}{E^{+}}$.}
\Ensure {$A$ is the transitive closure of $G$.}
\Function{UpdateTransitiveClosure}{$G, A, E^{+}$}

\State{$Q := $ empty queue}
\For{$uw \in E^{+}$}
    \If{$A_{u,w} = 0$}
        \State{$A_{u,w} := 1$}
        \State{$Q.\textsc{Enqueue}(uw)$}
    \EndIf
\EndFor

\While{$Q$ is not empty}
    \State $ab := Q.\textsc{Dequeue}$
    \For {$x \in Out(b) \cap \unreachable(a, A)$}\label{alg:inc-tc:out}
        \State $A_{a,x} := 1$\label{l:set1}
        \State{$Q.\textsc{Enqueue}(ax)$}
    \EndFor
    \For {$x \in In(a) \cap \cannotreach(b, A)$}\label{alg:inc-tc:in}
        \State $A_{x,b} := 1$ \label{l:set2}
        \State{$Q.\textsc{Enqueue}(xb)$}
    \EndFor
\EndWhile
\EndFunction
\end{algorithmic}
\caption{\label{alg:inc-tc}The queue-based incremental transitive closure.}
\end{algorithm}
\begin{lemma}\label{lem:tc-correct}
Algorithm~\ref{alg:inc-tc} is correct.
\end{lemma}
\begin{proof}
  We first show that the algorithm sets $A_{p,q} = 1$ only if there is a path from $p$ to $q$ in $G$.
We use induction on $k$.
For each $k \geq 1$ we show that at the beginning of the $k$-th iteration of the while loop:
\begin{enumerate}
  \item $A_{p,q} = 1$ only if there is a path from $p$ to $q$ in $G$,
  \item $Q$ only contains pairs $pq$ such that $A_{p,q} = 1$.
\end{enumerate}
The second item follows directly from the pseudocode.
  To show the first one, observe that the algorithm sets $A_{a,x} = 1$ in line~\ref{l:set1} when $A_{a,b} = 1$ (since $ab$ has been removed from $Q$) and $bx \in E(G)$.
  Similarly, it sets $A_{x,b} = 1$ when $A_{a,b} = 1$ and $xa \in E(G)$.
In both cases the induction hypothesis implies that the underlying paths exist.
This completes the first part of the proof.

  It remains to show that if adding $E^{+}$ causes $q$ to be reachable from $p$, then the algorithm sets $A_{p,q} = 1$ (we assume that $q$ was not reachable from $p$ before adding $E^{+}$).
Let $P = \dirpath{p}{q}{G}$.
We say that an edge $e$ of $P$ has a detour avoiding $E^{+}$ if a subpath of $P$ containing $e$ can be replaced with a path (detour) that does not contain edges of $E^{+}$ and the obtained path still connects $p$ and $q$.
As long as $P$ has an edge of $E^{+}$ that has a detour avoiding $E^{+}$ we replace its respective subpath with the detour.
This process terminates, as at each step we reduce the number of edges of $E^{+}$ on $P$.
Once the process terminates, $P$ contains some edge $uw$ of $E^{+}$, as there is no path from $p$ to $q$ that does not contain edges of $E^{+}$.
Moreover, $uw$ does not have a detour avoiding $E^{+}$.

Let $P' = \dirpath{p'}{q'}{G}$ be a subpath of $P$ that contains the edge $uw$.
  We have that $A_{p',q'} = 0$ before the algorithm was run.
  We now use induction on the length of $P'$ to show that the algorithm sets $A_{p',q'} = 1$.

The case when $k=1$ is easy, because the only valid path $P'$ of length $1$ contains the edge $uw$ by itself.
Let us now consider a path $P'$ of length $k > 1$.
Either the first or the last edge of $P'$ is not equal to $uw$.
Without loss of generality, let us assume that it is the last one and denote it by $rq'$.
  By the induction hypothesis, the algorithm sets $A_{p',r} = 1$.
This means that it then adds $p'r$ to the queue.

Consider the iteration of the while loop when $p'r$ is removed from the queue, that is $ab = p'r$.
Clearly, $q' \in Out(b) = Out(r)$, as $rq'$ is an edge on a path in $G$.
There are two cases to consider.
  If $q' \not\in \unreachable(a, A) = \unreachable(p', A)$, then $A_{p',q'} = 1$, so the procedure has already set $A_{p',q'} = 1$.
  Otherwise, $q' \in \unreachable(a, A)$ and then $A_{p',q'} = 1$ is set in line~\ref{l:set1} of some iteration of the for-loop.
This completes the proof.
\end{proof}

The following lemma highlights which part of Algorithm~\ref{alg:inc-tc} can be sped up.
\begin{lemma}\label{lem:inc-tc-time}
The running time of Algorithm~\ref{alg:inc-tc}, excluding the time needed to compute the sets $Out(b) \cap \unreachable(a, A)$ and $In(a) \cap \cannotreach(b, A)$, is $O(|E^{+}|+k+1)$, where $k$ is the number of entries of $A$ that are changed from $0$ to $1$.
\end{lemma}
\begin{proof}
Let $k$ be the total number of entries of $A$ that are changed from $0$ to $1$ by the algorithm.
The initial for loop takes $O(|E^{+}|)$ time.
Observe that just before the
algorithm sets $A_{a,x} := 1$ in line~\ref{l:set1}, we have $A_{a,x} = 0$, as $x \in \unreachable(a, A)$.
The similar reasoning applies to line~\ref{l:set2}.
Thus, the total running time of the two for loops is proportional to~$k$.

Moreover, the number of elements added to the queue $Q$ is at most $k$.
This implies that the total number of iterations of the while loop is at most $k$, and the lemma follows.
\end{proof}
\subsection{Proof of Theorem~\ref{thm:trans}}
Our goal is to maintain the transitive closure of a matrix that is given as a bitwise OR of matrices $A_1$ and $A_2$.
The following notation will prove useful in the remainder
of the paper.
\begin{definition}\label{def:mat-graph}
Let $M$ be some binary
matrix with rows and columns indexed with a vertex set $V_M$.
  Let $E_M=\{uv : {u,v\in V_M} \land M_{u,v}=1\}$.
We denote by $G(M)$ the graph $(V_M,E_M)$.
\end{definition}
Using this notation, our goal is to maintain the transitive closure of a graph $H = G(A_1) \cup G(A_2)$.
The vertex set of $H$ is $U_1 \cup U_2$.
Note that as the entries of $A_1,A_2$ change only from zero to one, the set $E(H)$ undergoes incremental updates.

In order to handle each update efficiently, we use Algorithm~\ref{alg:inc-tc}.
Moreover, we leverage the special structure of these matrices
to reduce the total time spent on computing the sets
$Out(b) \cap \unreachable(a,A)$ and $In(a) \cap \cannotreach(b, A)$ in lines~\ref{alg:inc-tc:out} and~\ref{alg:inc-tc:in}.
In the following we only deal with the former set; in order to compute the latter,
one needs to proceed symmetrically.

We first describe how the algorithm represents the matrices $A_1$ and $A_2$ and updates their representation.
Let us focus on $A_1$, as the representation of $A_2$ is analogous.
The algorithm uses Lemma~\ref{lem:monge_reach} to compute the partition $\aparts_1 = \{A^{S_1, T_1}, \ldots, A^{S_k, T_k}\}$.
For each $A^{S,T}\in\aparts_1$ it maintains the following:
\begin{itemize}
 \item the contents of the matrix $A^{S,T}$,
 \item the set $\actcols{A^{S,T}}$,
 \item for each row $s$ of $A^{S,T}$, the set
   ${Out_s(A^{S,T}) = \{ x : A^{S,T}_{s, x} = 1 \}}$. 
\end{itemize}
Let $T' =\actcols{A^{S,T}}$.
By Lemma~\ref{lem:monge_reach}, the ones in each row of $A^{S,T'}$ form $O(1)$ contiguous blocks.
Thus, internally the set $Out_s(A^{S,T})$ is represented by a constant number of intervals.
We have that $A^{S,T}_{s, x} = 1$ iff $x \in \actcols{A^{S,T}}$ and $x$ belongs to one of the intervals.
\begin{lemma}\label{lem:representation_a12}
The representation of matrices of $A_1$ and $A_2$ can be computed and maintained in $O(m^2 \log m)$ total time throughout any sequence of updates to $A_1$ and $A_2$.
\end{lemma}
\begin{proof}
For simplicity, let us consider $A_1$.
By Lemma~\ref{lem:monge_reach}, we can compute $\aparts_1$ in $O(m^2)$ time.
Observe that, by Lemma~\ref{lem:monge_reach}, this partition is independent of the values in the cells of $A_1$, so it does not need to be updated.

Whenever a cell of $A_1$ changes from $0$ to $1$ we update the corresponding cell in some $A^{S_i, T_i}$.
Note that, again by Lemma~\ref{lem:monge_reach}, for each cell of $A_1$ there is exactly one $A^{S,T} \in \aparts_1$ that contains the corresponding cell.
We can compute the mapping between the cells while computing $\aparts_1$, so the values in the matrices from $\aparts_1$ can be updated in constant time.
While updating the matrices, it is also easy to update the set $\actcols{A^{S,T}}$ for each $A^{S,T} \in \aparts_1$.

It remains to describe how to maintain $Out_s(A^{S,T})$ for each row $s$ of $A^{S,T}$.
Consider a sequence obtained from this row by removing cells in inactive columns.
We store this sequence in a balanced binary tree, keyed by the column.
This tree can be updated in $O(\log m)$ time when a column becomes active or some cell changes its value from $0$ to $1$.
Moreover, it can be easily augmented so as
to maintain the $O(1)$ intervals that we use to represent $Out_s(A^{S,T})$.

There is a small technicality here.
Each update to $A_1$ may consist of multiple cells changing value from $0$ to $1$ and only after all cells have been updated do the structural properties of Lemma~\ref{lem:monge_reach} hold.
Thus, before all cells in the considered row are updated, the representation of $Out_s(A^{S,T})$ may consist of more than a constant
number of intervals, which could affect the running time.
However, this does not pose an issue, as we can force the binary tree to compute only the first $O(1)$ intervals of the representation.
As a result, the running time is not affected and the representation is correct once all entries of $A_1$ are updated.

Updating the matrices and the sets $\actcols{A^{S,T}}$ takes time which is linear in the size of $A_1$, that is $O(m^2)$.
To bound the running time of updating the sets $Out_s(A^{S,T})$ we
note that each element is inserted into the binary tree only once (when the corresponding column becomes active) and can be updated only once (when the corresponding matrix entry changes value from $0$ to $1$).
Thus, we spend $O(\log m)$ time for each cell of each $A^{S,T} \in \aparts_1$.
This gives $O(m^2 \log m)$ time in total, and yields the lemma.
\end{proof}

For each $a\in U_1\cup U_2$ and $A^{S,T} \in \aparts_1 \cup \aparts_2$ the algorithm stores a \emph{reachability candidates set} $\cset_a(A^{S,T}) \subseteq \actcols{A^{S,T}}$.
The algorithm maintains the invariant that
$\cset_a(A^{S,T}) =\linebreak \actcols{A^{S,T}} \cap \unreachable(a, A)$.
Whenever a column $t \in T$ of $A^{S,T}$ becomes active, we check whether $A_{a,t} = 0$ and if this is the case, we insert $t$ into $\cset_a(A^{S,T})$.
Once we set $A_{a,t} = 1$, $t$ is removed from $\cset_a(A^{S,T})$ (and, clearly, never added again, as $A_{a,t}$ will never change its value back to $0$).
Each $\cset_a(A^{S,T})$ is stored as a dynamic predecessor data structure, such as the \emph{van Emde Boas tree} \cite{vEB:1977}.
\begin{lemma}\label{lem:reachability_cand}
The sets of reachability candidates can be updated in $O(m^2 \log m \log \log m)$ total time.
\end{lemma}
\begin{proof}
Fix $a \in U_1 \cup U_2$.
For each $A^{S,T} \in \aparts_1 \cup \aparts_2$ each element of $T$ is added to and removed from $\cset_a(A^{S,T})$ at most once.
If $\aparts_1 \cup \aparts_2 = \{A^{S_1, T_1}, \ldots, A^{S_k, T_k}\}$ then the total number of operations is $\sum_{i=1}^k |T_i|$.
By Lemma~\ref{lem:monge_reach}, for each $b \in U_1 \cup U_2$ there are $O(\log m)$ matrices $A^{S,T} \in \aparts_1 \cup \aparts_2$ such that $b \in T$.
Thus we make at most $\sum_{i=1}^k |T_i| = O(m \log m)$ operations for a fixed $a$, which gives $O(m^2 \log m)$ operations in total.
Since each operation on a van Emde Boas tree takes $O(\log \log m)$ time, we conclude that maintaining reachability candidates takes $O(m^2 \log m \log \log m)$ time.
\end{proof}

We are now ready to show how to speed up line~\ref{alg:inc-tc:out} of Algorithm~\ref{alg:inc-tc}.
The goal is to compute the set $Out(b) \cap \unreachable(a,A)$ efficiently.
Algorithm~\ref{alg:inc-tc} traverses this set, once it is computed, but this does not affect the running time considerably.
Only the computation of the set could be slow.
This means that it suffices to compute the set in time which is, say, almost linear in its size.
Our algorithm, roughly speaking, uses the property that $Out(b)$ is represented by a small number of intervals, so computing an intersection with the set $Out(b)$ is easy.
Recall that $\prowmap_b(\aparts_i)$ is the subset of $\aparts_i$ consisting of matrices which contain the row $b$.
\begin{lemma}\label{lem:compute_out}
Let $$Q=\bigcup_{i\in\{1,2\}}\bigcup_{A^{S,T}\in\prowmap_b(\aparts_i)} \big (\cset_a(A^{S,T})\cap Out_b(A^{S,T})\big ).$$
Then $Q = Out(b) \cap \unreachable(a,A)$.
\end{lemma}
\begin{proof}
By Lemma~\ref{lem:monge_reach} and the definition of the sets $Out_b(A^{S,T})$, we have
  \begin{align*}
    Out(b)&=\bigcup_{i\in\{1,2\}}\bigcup_{A^{S,T}\in\prowmap_b(\aparts_i)} Out_b(A^{S,T})\\
    &=\bigcup_{i\in\{1,2\}}\bigcup_{A^{S,T}\in\prowmap_b(\aparts_i)}\actcols{A^{S,T}}\cap Out_b(A^{S,T}).
  \end{align*}
Moreover $\cset_a(A^{S,T}) = \actcols{A^{S,T}} \cap \unreachable(a, A)$.
Hence
\begin{align*}
Q&=\bigcup_{i\in\{1,2\}}\bigcup_{A^{S,T}\in\prowmap_b(\aparts_i)}\cset_a(A^{S,T})\cap Out_b(A^{S,T})\\
  &=\bigcup_{\substack{i\in\{1,2\}\\
  A^{S,T}\in\prowmap_b(\aparts_i)}} \actcols{A^{S,T}} \cap \unreachable(a, A) \cap Out_b(A^{S,T})\\
    & = Out(b) \cap \unreachable(a, A).\qedhere
\end{align*}
\end{proof}

\begin{lemma}\label{lem:compute_q}
$Q = Out(b) \cap \unreachable(a,A)$ can be computed in $O((|Q| + \log m)\log \log m)$ time.
\end{lemma}
\begin{proof}
By Lemma~\ref{lem:compute_out}, it suffices to compute $$\bigcup_{i\in\{1,2\}}\bigcup_{A^{S,T}\in\prowmap_b(\aparts_i)}\cset_a(A^{S,T})\cap Out_b(A^{S,T}).$$
Since $|\prowmap_b(\aparts_i)| = O(\log m)$ (by Lemma~\ref{lem:monge_reach}), this is a sum of $O(\log m)$ sets of the form $\cset_a(A^{S,T})\cap Out_b(A^{S,T})$.
  For each fixed $i$, these sets are disjoint since for the matrices $A^{S,T}$ such that $A^{S,T}\in\prowmap_b(\aparts_i)$, the sets $T$ are disjoint.

Let us focus on computing $\cset_a(A^{S,T})\cap Out_b(A^{S,T})$.
Recall that $Out_b(A^{S,T})$ is represented as a constant number of intervals, such that $x \in Out_b(A^{S,T})$ iff $x$ belongs both to $\actcols{A^{S,T}}$ and one of the intervals.
However, we also have $\cset_a(A^{S,T}) \subseteq \actcols{A^{S,T}}$.
As a result, to compute the intersection it suffices to take elements of $\cset_a(A^{S,T})$ that are contained in the intervals describing $Out_b(A^{S,T})$.
Since $\cset_a(A^{S,T})$ is represented as a van Emde Boas tree, finding a single element of $\cset_a(A^{S,T})$ contained in a given interval can be done in $O(\log \log m)$ time.
Thus, we spend $O(\log \log m)$ time on computing each element of $Q$ plus $O(\log m\log \log m)$ time to consider $O(\log m)$ sets that comprise the sum.
\end{proof}
\begin{proof}[Proof of Theorem~\ref{thm:trans}]
By Lemmas~\ref{lem:representation_a12} and~\ref{lem:reachability_cand} the total cost of computing and updating the representations of $A_1$ and $A_2$ and the reachability candidates is $O(m^2 \log m \log \log m)$.
By Lemma~\ref{lem:inc-tc-time}, the total running time of Algorithm~\ref{alg:inc-tc}
excluding the cost of lines~\ref{alg:inc-tc:out}~and~\ref{alg:inc-tc:in} is $O(m^2)$.
By Lemma~\ref{lem:compute_q}, it takes $O((q + \log m)\log \log m)$ to execute each of these lines, assuming that they compute a set of size $q$.
Since each such set is then traversed by the algorithm, summing this overhead over all queue entries gives $O(m^2 \log m\log \log m)$ total time.
Thus, the overall running time is $O(m^2 \log m \log \log m)$.
\end{proof}

%% file: dag-trans.tex
\newcommand{\orig}{\ensuremath{\mathrm{orig}}}
\newcommand{\corr}{\ensuremath{\mathcal{S}}}

\section{The Switch-On Reachability Data Structure}\label{sec:dag-trans}
In this section we show how to combine the recursive decomposition
tree with the switch-on transitive closure data structure of Theorem~\ref{thm:trans}
in order to solve efficiently the switch-on reachability problem.
By Lemmas~\ref{lem:switch-inter}~and~\ref{lem:inter-scc}, this will imply
an algorithm for solving the decremental strongly-connected components
problem within the same time bounds.

Let $G=(V,E)$ be a plane embedded digraph with $O(n)$ vertices and $O(n)$ edges.
The following technical lemma is proved in Section~\ref{sec:nonsimple}.
\begin{lemma}\label{lem:simple-ext}
  Let $G=(V,E)$ be a planar digraph with $O(n)$ vertices and $O(n)$ edges.
  In $O(n\log{n})$ time we can construct a plane digraph $G'=(V',E')$ along with:
  \begin{itemize}
  \item a simple recursive decomposition $\rtree(G')$,
  \item two disjoint subsets $E_0, E_1\subseteq E'$,
  \item a bijective mapping $\corr$ from $V$ to the set of strongly connected components of $(V',E_0)$,
  \item a bijective function $p:E\to E_1$,
  \end{itemize}
  satisfying the following properties.
  \begin{enumerate}
  \item $|V'|=O(n)$ and $|E'|=O(n)$.
  \item There are no inter-SCC edges in $(V',E_0)$.
  \item For any $uv=e\in E$, let $u'v'=p(e)\in E_1$. Then $u'\in \corr(u)$ and $v'\in \corr(v)$.
  \end{enumerate}
\end{lemma}
\begin{corollary}\label{col:subset-equiv}
  Let $G=(V,E)$. Let $G'=(V',E')$, $E_0$, $E_1$, $\corr$ and $p$ be defined as in Lemma~\ref{lem:simple-ext}.
  Let $u',v'\in V'$ and let $u'\in\corr(u)$ and $v'\in\corr(v)$.
  Then for any $F\subseteq E$, a path $\dirpath{u'}{v'}{}$ exists
    in $(V',E_0\cup p(F))$ if and only if a path $\dirpath{u}{v}{}$
    exists in $(V,F)$.
\end{corollary}
\begin{proof}
  The proof is trivial when $u=v$, so suppose $u\neq v$.

  Let $\pi=v_0,e_1,v_1,\ldots,v_{k-1},e_{k},v_k$ be a simple path in $(V',E_0\cup p(F))$, where $v_0=u'$
  and $v_k=v'$.
  For $x\in V'$, denote by $s(x)$ the unique vertex of $V$ such that $x\in \corr(s(x))$.
  Let $e_{i_1},\ldots,e_{i_\ell}$ be the edges of $E(\pi) \cap p(F)$ in order
  of their occurrence on $\pi$.
  As for all edges $v_{j}v_{j+1}$ of $E_0$ lying on $\pi$ we have $s(v_j)=s(v_{j+1})$
  (recall that $(V',E_0)$ has no inter-SCC edges),
  we conclude that $s(v_{i_q-1})=s(v_{i_{q-1}})$ for any $q=2,\ldots,\ell$.
  We also have $s(v_{i_1-1})=s(v_0)=s(u')$ and $s(v_{i_\ell})=s(v_k)=s(v')$.
  Moreover, for each $i_j$, $p^{-1}(e_{i_j})=s(v_{i_j-1})s(v_{i_j})$.
  Thus, $p^{-1}(e_{i_1})p^{-1}(e_{i_2})\ldots p^{-1}(e_{i_\ell})$ is a path
  $\dirpath{s(u')}{s(v')}{}$ in $(V,F)$.
  To conclude, note that $s(u')=u$ and $s(v')=v$.
  
  Conversely, let $e_1\ldots e_k$ be a path $\dirpath{u}{v}{}$ in $(V,F)$, where $e_i=u_iv_i$.
  Let $u_i',v_i'$ be any vertices of $V'$
  such that $u_i'\in\corr(u_i)$ and $v_i'\in\corr(v_i)$, choosing $u_1'=u'$ and $v_k'=v'$.
  For each $i=2,\ldots,k$ we have $v_{i-1}=u_i$ and thus
  there exists a path $\dirpath{v_{i-1}'}{u_i'}{}$ in $(V',E_0)$.
  We now show that for each $i$ there is also a path
  $\dirpath{u_i'}{v_i'}{}$ in $(V',E_0\cup\{p(e_i)\})$.
  Let $p(e_i)=x_iy_i$.
  Note that $s(x_i)=s(u_i')$ and $s(y_i)=s(v_i')$
  and hence there is a path $u_i'\to x_i\to y_i\to v_i'$ in $(V',E_0\cup\{p(e_i)\})$.
  
  We conclude that there is a path $u_1'\to v_1'\to u_2'\to v_2'\to \ldots u_k'\to v_k'$
  in $(V',E_0\cup p(F))$.
\end{proof}
The data structure presented in this section requires a simple recursive decomposition $\rtree(G)$.
We now prove that the general case of any planar digraph $G$ can be reduced
to the case when we are given a simple recursive decomposition without increasing the overall asymptotic
running time of the whole algorithm.
\begin{lemma}\label{lem:simple-reduction}
Let $G$ be a planar digraph with $O(n)$ vertices and $O(n)$ edges.
In $O(n\log{n})$ time we can reduce the switch-on reachability problem on $G$
to the switch-on reachability problem on a plane graph $G'$ with a given
  simple recursive decomposition $\rtree(G')$ and such that $|V(G')|=O(n)$ and $|E(G')|=O(n)$.
\end{lemma}
\begin{proof}
  Let $G'=(V',E')$ and $p$ be defined as in Lemma~\ref{lem:simple-ext}.
  By definition, $G'$ is accompanied with a simple recursive decomposition $\rtree(G')$.
  Hence, we only need to show how to translate the updates to $G$ into the updates to $G'$.

  We first switch on all edges of $E_0$. When an edge $e\in E$ of $G$ is switched on,
  we switch on $p(e)$ in $G'$.
  The edges of $E'\setminus(E_0\cup E_1)$ are never switched on in $G'$.
  Let $F\subseteq E$ be the subset of edges of $G$ that are switched on at some point of time.
  For any $uv=e\in E$, let $u'v'=e'=p(e)$.
  By Corollary~\ref{col:subset-equiv}, the path $v\to u$ exists in $(V,F)$ iff the path
  $v'\to u'$ exists in $(V',E_0\cup p(F))$.
  Thus, to track the reachability between the endpoints of $e$, we only need
  to track the reachability between the endpoints of $p(e)$.
  This is done by solving the switch-on reachability problem on $G'$.
\end{proof}
In the remaining part of this section we assume that we are given
a simple recursive decomposition $\rtree(G)$ of $G$.

For $H\in\rtree(G)$, denote by $G-H$ the edge-induced subgraph
$G[E\setminus E(H)]$.
\begin{lemma}\label{lem:boundary-diff}
Let $H\in\rtree(G)$. Then $V(H)\cap V(G-H)\subseteq \bnd{H}$.
\end{lemma}
\begin{proof}
The proof is by induction of the level of $H$ in $\rtree(G)$. If $H$ is the root, then $G-H$ is empty
and the statement is clearly true.

Consider now $H$ which is not the root of $\rtree(G)$.
Let $P$ be the parent of $H$ and assume $V(P)\cap V(G-P)\subseteq \bnd{P}$.
Let $S$ be the sibling of~$H$.
  Recall that, by definition, $\bnd{H}=(V(H)\cap\bnd{P})\cup(V(H)\cap V(S))$.
  Let $v\in V(H)\cap V(G-H)$.
  Clearly, $v\in V(P)$ since $E(H)\subseteq E(P)$.
  Moreover, $v$ is incident to at least one edge $e$ of $E\setminus E(H)$.
  If $e\in E\setminus E(P)$, then $v\in V(G-P)$
  and thus $v\in\bnd{P}\cap V(H)$.
  Otherwise, $e\in E(S)$ and hence $v\in V(S)\cap V(H)$.
\end{proof}

\begin{lemma}\label{lem:hole-curve}
  Let $H\in\rtree(G)$. A cycle bounding a hole of $H$ is a separator curve of $H$.
\end{lemma}
\begin{proof}
  By Fact~\ref{fac:face-curve}, the cycle bounding a simple face of $H$ is
  a separator curve of $H$.
\end{proof}
\begin{lemma}\label{lem:hole-curve-diff}
  Let $H\in\rtree(G)$. A cycle bounding a hole of $H$ is a separator curve of $G-H$.
\end{lemma}
\begin{proof}
  Note that the holes of $H$ can be seen as unions of original faces of $G$, merged
  by removing the edges of $G-H$ from $G$.
  Thus, each edge of $G-H$ lies inside some unique hole $h$ of $H$.
  Let $E_h\subseteq E(G-H)$ be the subset of edges lying inside $h$.

  Let $\curv_h$ be the cycle bounding a hole $h$ of $H$.
  Assume wlog. that $h$ is a bounded face, then the inside of $\curv_h$ is
  the same as the inside of $h$ (the case when $h$ is unbounded
  is analogous; we replace each occurrence of
  ``inside of $\curv_h$'' with ``outside of $\curv_h$'').
  By Lemma~\ref{lem:hole-curve}, for each $\curv_h$, $H$ lies
  weakly outside $\curv_h$.
  If $h$ is the only hole of $H$, then clearly $G-H=G[E_h]$ lies weakly
  on one side of $\curv_h$ and the Lemma holds.

  Assume now that $h$ is not the only hole, and let $h'\neq h$ be some other hole of $H$.
  As $h'$ and $h$ are disjoint, $h'$ lies strictly outside
  $\curv_h$.
  Thus, all edges of $E_{h'}$ lie strictly outside $\curv_h$.
  Hence, $V(G[E_h])\cap V(G[E_{h'}])=\emptyset$ for $h'\neq h$.
  To conclude, note that for each weakly connected component of $G-H$,
  the edges of that component are contained in a unique subset $E_h$.
\end{proof}
\begin{remark}
The assumption that $\rtree(G)$ is simple is crucial to proving Lemma~\ref{lem:hole-curve-diff},
which does not hold if the holes of $H$ are not pairwise-disjoint or not necessarily
simple.
\end{remark}

\begin{lemma}\label{lem:boundary-curves}
Let $H\in\rtree(G)$. The set $\bnd{H}$ can be partitioned into $O(1)$ sets
  $\bnd_1{H},\ldots,\bnd_\ell{H}$ so that each $\bnd_i{H}$ lies on a curve $\curv_i$
such that $\curv_i$ is a separator curve of both $H$ and $G-H$.
\end{lemma}
\begin{proof}
  Let $h_1,\ldots,h_\ell$ be the holes of $H$.
  By Lemmas~\ref{lem:hole-curve}~and~\ref{lem:hole-curve-diff}, we can set $\bnd_i{H}$
  to be the subset of $\bnd{H}$ lying on $h_i$.
  By the definition of $\rtree(G)$, each $v\in\bnd{H}$ lies on a unique hole of $H$.
\end{proof}

Our strategy will be to maintain, for each \emph{leaf} subgraph $H\in\rtree(G)$,
a binary 
matrix $\gtc(H)$ with both rows and columns indexed with the vertices of $H$, such that
$\gtc(H)_{u,v}=1$ iff there exists a path $\dirpath{u}{v}{\swon{G}}$.
Recall that our goal is to track, for each $uv\in E$,
the information whether there is a path $\dirpath{v}{u}{\swon{G}}$ consisting of edges
that are switched on.
By the definition of $\rtree(G)$, each edge $uv\in E$ is contained in
some leaf subgraph $H\in\rtree(G)$ and thus all the needed information
is contained in the matrices $\gtc(\cdot)$.
To efficiently update the matrices $\gtc(\cdot)$ while the edges are switched on,
we maintain two types of auxiliary information for each $H\in\rtree(G)$: 
\begin{enumerate}
  \item 
  A binary 
  matrix $\intc(H)$ with both rows and columns indexed
    with the vertices of $\bnd{H}$, such that $\intc(H)_{u,v}=1$ iff
    $u,v\in\bnd{H}$ and there exists a path $\dirpath{u}{v}{}$ in $\swon{H}$;
  \item 
   A binary 
  matrix $\extc(H)$ with both rows and columns indexed
    with the vertices of $\bnd{H}$, such that $\extc(H)_{u,v}=1$ iff
    $u,v\in\bnd{H}$ and there exists a path $\dirpath{u}{v}{}$ in $\swon{G-H}$.
\end{enumerate}
Additionally, for each leaf subgraph $H$ we maintain the transitive closure
of $\swon{H}$ in a binary
matrix $\leaftc(H)$.
Note that $\intc(H)$ is a submatrix of $\leaftc(H)$, namely $\intc(H) = \leaftc(H)[\bnd{H}]$.
In the following we discuss how the matrices $\intc(\cdot)$, $\extc(\cdot)$,
$\leaftc(\cdot)$ and $\gtc(\cdot)$ interplay and we show how they can be efficiently
maintained for any sequence of edge switch-ons.
Observe that all these matrices undergo monotone changes: the edge switch-ons
can only cause their entries to change from $0$ to $1$.
Therefore, an $m\times m$ matrix can be updated only $O(m^2)$ times.
\begin{lemma}\label{lem:leafcomp}
Let $H$ be a leaf subgraph of $\rtree(G)$. $\leaftc(H)$ can be initialized and maintained in
$O(1)$ total time subject to any sequence of switch-ons of edges of $H$.
\end{lemma}
\begin{proof}
  As each leaf subgraph has constant size, once one of the $O(1)$ edges of $H$ is switched-on,
the matrix $\leaftc(H)$ can be recomputed from scratch in $O(1)$ time.
\end{proof}
In the following lemma we say that a matrix $M$ depends only on some matrices $M_1, \ldots, M_k$ if the information contained in
$M_1, \ldots, M_k$ is sufficient to compute $M$.

\begin{lemma}\label{lem:summatrices}
  Let $H$ be a subgraph of $\rtree(G)$.
\begin{enumerate}
\item If $H$ is a leaf subgraph, then $\intc(H)$ depends only on $\leaftc(H)$.
\item If $H$ is a non-leaf subgraph, then $\intc(H)$ depends only on $\intc(\child_1(H))$ and $\intc(\child_2(H))$.
\item If $H$ is a non-root subgraph, then let $P$ and $S$ be the parent and the sibling of $H$ in $\rtree(G)$, respectively.
  $\extc(H)$ depends only on $\extc(P)$ and $\intc(S)$.
\item If $H$ is a leaf subgraph of $\rtree(G)$, then $\gtc(H)$ depends only on
  $\leaftc(H)$ and $\extc(H)$.
\end{enumerate}
  The matrices in parts (1) and (4) can be maintained in $O(1)$ total time.
  The matrices in parts (2) and (3) can be maintained in
$O(m^2\log{m}\log\log{m})$
  total time subject to any sequence of updates to the matrices they depend on,
  where $m=|\bnd{\child_1(H)}\cup\bnd{\child_2(H)}|$ in part (2)
  and $m=|\bnd{P}\cup\bnd{S}|$ in part (3).
\end{lemma}
The dependencies stated in Lemma~\ref{lem:summatrices}
are depicted in Figure~\ref{fig:dep}.
\begin{proof}
  \emph{(1)}\hspace{2pt}
  $\intc(H)$ is a submatrix of $\leaftc(H)$, which is of constant size.
  The total number of updates to $\leaftc(H)$ is $O(1)$.
  
  \emph{(2)}
  For $i=1,2$, set $G_i=\child_i(H)$, $U_i=\bnd{\child_i(H)}$, and ${A_i=\intc(\child_i(H))}$.
  Note that by Lemma~\ref{lem:boundary-curves}, the set $U_i$ lies on a constant number of
  separator curves of $G_i$.
  We also have ${V(G_1)\cap V(G_2)}\subseteq {U_1\cap U_2}$ and $|U_1\cup U_2|=m$.
  By Theorem~\ref{thm:trans}, the reachability matrix $A$ of $U_1\cup U_2$ in
  $\swon{G_1}\cup\swon{G_2}$ can be maintained in $O(m^2\log{m}\log\log{m})$
  total time, subject to any sequence of updates to $A_1$ and $A_2$.
  As $\bnd{H}\subseteq\bnd{\child_1(H)}\cup\bnd{\child_2(H)}$,
  $\intc(H)$ is a submatrix of $A$ and thus can be maintained within
  the same time bounds.

  \emph{(3)}
Set $G_1=G-P$, $G_2=S$, $U_1=\bnd{P}$, $U_2=\bnd{S}$, $A_1=\extc(P)$ and $A_2=\intc(S)$.
  Note that by Lemma~\ref{lem:boundary-diff}, ${V(S)\cap V(G-P)}\subseteq {V(P)\cap V(G-P)}\subseteq \bnd{P}$
  and similarly $V(S)\cap V(G-P)\subseteq V(S)\cap V(G-S)\subseteq \bnd{S}$.
  Therefore, $V(G_1)\cap V(G_2)\subseteq U_1\cap U_2$.
  By Lemma~\ref{lem:boundary-curves}, the set $U_i$ lies on a constant number of
  separator curves of $G_i$.

  By Theorem~\ref{thm:trans}, the matrix $A$ of ${U_1\cup U_2}$ in
$\swon{G_1}\cup\swon{G_2}$ can be maintained in $O(m^2\log{m}\log\log{m})$ time,
  since $|U_1\cup U_2|=m$.
  As $E(H)\cap E(S)=\emptyset$ and $E(H)\cup E(S)=E(P)$, $G_1\cup G_2=(G-P)\cup S=G-H$.
  Thus, $A$ actually represents the reachability between vertices $\bnd{P}\cup\bnd{S}$ in $\swon{G-H}$.
  As clearly $\bnd{H}\subseteq\bnd{P}\cup\bnd{S}$, $\extc(H)$ is a submatrix
  of $A$ and thus can be maintained in $O(m^2\log{m}\log\log{m})$ time as well.

  \emph{(4)}
  We show that $\gtc(H)_{u,v}=1$ iff there exists a path $\dirpath{u}{v}{}$ in
  the graph $T_H=G(\leaftc(H))\cup G(\extc(H))$ (see Definition~\ref{def:mat-graph}).
  Note that $V(T_H)=V(H)$.

  Clearly, by the definitions of the matrices $\leaftc(H)$ and $\extc(H)$,
  if there exists a path $\dirpath{u}{v}{T_H}$, there also exists a path $\dirpath{u}{v}{\swon{G}}$.

  Now, let $P$ be some path $\dirpath{u}{v}{\swon{G}}$ such that $u,v\in V(H)$.
  Split $P$ into maximal subpaths $P_1,\ldots,P_k$ fully contained (as far as their edges are concerned)
  in either $E(H)$ or $E(G-H)$.

  Let $P_i$ be a subpath $\dirpath{a}{b}{}$ entirely contained in $G-H$.
  Then $a,b\in V(G-H)$.
  We prove that $a\in V(H)$. The proof that $b\in V(H)$ is analogous.
  We either have $i=1$ and $a=u$ implies $a\in V(H)$, or the path $P_{i-1}$ is fully contained in $H$
  and thus $a\in V(H)$ as well.
  Now, as $a,b\in V(G-H)\cap V(H)$, by Lemma~\ref{lem:boundary-diff},
  we have $a,b\in\bnd{H}$.
  Hence, there is an edge $ab$ in $G(\extc(H))$.

  By the definition of $\leaftc(H)$, for $P_i=\dirpath{a}{b}{H}$,
  there is an edge $ab$ in $G(\leaftc(H))$.
  We conclude that for each subpath $P_i= \dirpath{a}{b}{\swon{G}}$, there is an edge $ab\in E(T_H)$ and
  thus there is a path $\dirpath{u}{v}{T_H}$.

  As each of the matrices $\leaftc(H),\extc(H)$ is of
  constant size and $\ell=O(1)$, the matrix $\gtc(H)$ encoding the transitive closure of $T_H$
  can be recomputed in $O(1)$ time after any update to these matrices.
  Moreover, the total number of changes to the matrices $\leaftc(H),\extc(H)$
  is constant.
\end{proof}

\begin{figure}[hbt!]
\centering
  \includegraphics{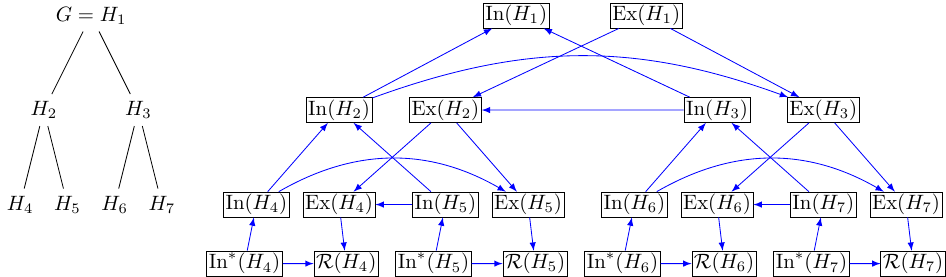}
\caption{An example decomposition $\rtree(G)$ (top) and the dependencies between
the corresponding matrices $\intc(H),\leaftc(H),\extc(H),\gtc(H)$, for $H\in\rtree(G)$ (bottom).}
\label{fig:dep}
\end{figure}
\newcommand{\deplist}{\ensuremath{\mathcal{L}}}
Note that each matrix depends directly on at most two other matrices.
The dependencies allow us to organize the matrices
$\bigcup_{H\in\rtree(G)}\{\leaftc(H),\intc(H),\extc(H),\gtc(H)\}$
in a \emph{dependency list} $\deplist$,
such that each matrix depends only on matrices earlier in the sequence.
The order of the elements of $\deplist$ is as follows.
The matrices form groups according to their type.
The order on groups is $\leaftc, \intc, \extc, \gtc$.
Matrices within the groups $\leaftc$, $\intc$ and $\gtc$ are sorted
  increasingly inclusion-wise by their corresponding subgraphs of $\rtree(G)$.
On the other hand, the matrices within the group $\extc$ are sorted
  decreasingly inclusion-wise by their corresponding subgraphs of $\rtree(G)$.
\begin{lemma}\label{lem:simple-switch-on}
  Let $G$ be a plane digraph and let $n=|V(G)|$.
  Let $\rtree(G)$ be a simple recursive decomposition of $G$.
  The switch-on reachability problem on $G$
can be solved in $O(n\log^2{n}\log\log{n})$ time.
\end{lemma}
\begin{proof}
We initialize the data structures maintaining the matrices of the
dependency list.
When an edge $e$ is switched on, we create a priority queue $Q$
of matrices that potentially need updates.
The elements of $Q$ are keyed by their position in the list $\deplist$.
First, the unique matrix $\leaftc(H)$ such that $e\in E(H)$, is pushed to $Q$.
We repeatedly pop matrices $M$ out of $Q$ and process
either the edge switch-on (in the case of $\leaftc(\cdot)$ matrices)
or the changes in matrices $M'$ such that $M$ directly depends on $M'$.
If the matrix $M$ changes after switching $e$ on, we push
to $Q$ all matrices $M_2$ such that $M_2$ directly depends on $M$.
The correctness of this update procedure follows from the fact that the dependencies
  do not form cycles and all required matrices $\leaftc(\cdot)$ are notified about the switch-on.

Observe that as for each matrix $M\in\deplist$, there are at most
two matrices of $\deplist$ that depend directly on $M$,
the total number of times a matrix is inserted to $Q$
is proportional to the total size of the matrices of $\deplist$ plus
  the sum of sizes of the leaf subgraphs of $\rtree(G)$, i.e.,
  $O(\sum_{H\in\rtree(G)}(|\bnd{H}|^2+|E(H)|))=O(n\log{n})$,
  by Corollary~\ref{cor:rtree-bounds}.
Thus, the cost of priority queue operations is $O(n\log^2{n})$.
  By Lemma~\ref{lem:summatrices}, the cost of maintaining $\intc(H)$ for a non-leaf $H$
  is $O((|\bnd{\child_1(H)}|^2+|\bnd{\child_2(H)}|^2)\log{n}\log\log{n})$.
  Similarly, maintaining $\extc(H)$ for a non-root $H$ costs
  $O((|\bnd{P}|^2+|\bnd{S}|^2)\log{n}\log\log{n})$, where $P,S$ are its parent and sibling.
  Here we use $|X\cup Y|^2\leq 2|X|^2+2|Y|^2$.
  Each piece occurs as an input to only $O(1)$ of these computations,
  and the leaf computations in parts (1) and (4), together with maintaining $\leaftc(H)$, cost $O(n)$ in total.
  Recall that $\sum_{H\in\rtree(G)}|\bnd{H}|^2=O(n\log{n})$.
  Thus, the total cost over all subgraphs $H\in\rtree(G)$ is
$O(n\log^2{n}\log\log{n})$.
\end{proof}
\begin{theorem}
  Let $G$ be a planar digraph with $O(n)$ vertices and $O(n)$ edges. The switch-on reachability problem on $G$
can be solved in
$O(n\log^2{n}\log\log{n})$ time.
\end{theorem}
\begin{proof}
We first apply Lemma~\ref{lem:simple-reduction} to reduce the problem
to the case with a given simple recursive decomposition $\rtree(G)$.
The reduction takes $O(n\log{n})$ time and produces a graph with $O(n)$ vertices.
Then, we apply Lemma~\ref{lem:simple-switch-on}.
\end{proof}

%% file: 2conn.tex
\begin{definition}
  Let $G=(V,E)$ be a digraph. We call an edge $e\in E$
  a \emph{strong bridge} of $G$ iff $G-e$ has more strongly connected
  components than $G$.
\end{definition}

\begin{lemma}\label{lem:strong-bridge}
  Let $G$ be a plane digraph and let
  $uv=e\in E(G)$ be an intra-SCC edge of $G$. Let $ab=\dual{e}\in E(\dual{G})$.
  The edge $e$ is a strong bridge of $G$ iff there exists a path $\dirpath{a}{b}{}$ in $\dual{G}-\dual{e}$.
\end{lemma}
\begin{proof}
By Lemma~\ref{lem:inter-dual}, $\dual{e}$ is an inter-SCC edge of $\dual{G}$,
  or equivalently, there is no path $\dirpath{b}{a}{}$ in $\dual{G}$.
  Since $e$ is intra-SCC, it is not a bridge of the underlying undirected graph,
  so deleting $e$ merges its two incident faces and contracts $\dual{e}$,
  as in Fact~\ref{fact:remove-contr}, even if $G$ is disconnected.
  Note that $e$ is a strong bridge iff some edge $f\in E\setminus\{e\}$ is intra-SCC
  in $G$ but is inter-SCC in $G-e$.
  Equivalently, $\dual{f}$ is inter-SCC in $\dual{G}$
  and is intra-SCC in $\contr{\dual{G}}{\dual{e}}$.
  Consequently, there is a set of edges $C\subseteq E(\contr{\dual{G}}{\dual{e}})$, $\dual{f}\in C$, such that
  $C$ forms a directed cycle in $\contr{\dual{G}}{\dual{e}}$, but does not form a directed
  cycle in $\dual{G}$.
  Thus, $C$ is in fact a directed path $a\to b$ in $\dual{G}$.
  However, $\dual{e}\notin C$, and hence the lemma follows.
\end{proof}

\begin{lemma}\label{lem:dec-strong-bridge}
Let $G$ be a planar digraph and let $n=|V|$. The set of strong bridges
  of $G$ can be maintained subject to edge deletions in $O(n\log^2{n}\log\log{n})$
  total time.
\end{lemma}
\begin{proof}
  \newcommand{\gdr}{\ensuremath{G_+}}

  Similarly as in the proof of Lemma~\ref{lem:switch-inter},
  we reduce our problem to maintaining
  the graph $\gdr=(V(\dual{G}),E(\dual{G})\cup E^R(\dual{G}))$
  under edge switch-ons.
  Retaining inter-SCC edges as in that reduction also preserves strong bridges.
  A contraction of $e$
  in $\dual{G}$ is translated to a switch-on of $e^R$ in $\gdr$.
  By Lemma~\ref{lem:strong-bridge}, applied only to surviving intra-SCC edges, it is sufficient to solve
  the following extended switch-on reachability
  problem on the input graph $\gdr$.
  Given a planar digraph $H$,
  we need to maintain for each edge $uv=e\in E(H)$ whether there is
  a path $\dirpath{u}{v}{}$ in $\swon{H}-e$, where $\swon{H}$ is
  an incremental, switched-on subgraph of $H$.

  The next step is to reduce such an extended switch-on reachability
  problem to the case when the underlying graph has a simple recursive decomposition $\rtree(H)$.
  To this end, we apply Lemma~\ref{lem:simple-ext} to $H$
  and obtain the graph $H'=(V',E')$,
  function $p$ and the subsets $E_0,E_1\subseteq E'$.
  Analogously as in the proof of Lemma~\ref{lem:simple-reduction}, the edges $E_0$
  are switched on during the initialization, the edges $E'\setminus(E_0\cup E_1)$ are never switched-on,
  whereas a switch-on of the edge $e\in E(H)$ is translated to a switch-on
  of $p(e)$ in $H'$.
  Note that, by Corollary~\ref{col:subset-equiv}, for any $F\subseteq E(H)$, and any $uv=e\in H$, a path $\dirpath{u}{v}{}$
  exists in $(V(H),F\setminus\{e\})$ if and only if a path $\dirpath{u'}{v'}{}$ exists
  in $(V',E_0\cup p(F\setminus\{e\}))=(V',E_0\cup p(F)\setminus\{p(e)\})$,
  where $p(e)=u'v'$.
  Hence, in order to maintain the reachability between the endpoints of $e$ in $\swon{H}-e$
  we only need to maintain the reachability between the endpoints of $p(e)$ in $\swon{H'}-p(e)$.
  Recall that the cost of applying Lemma~\ref{lem:simple-ext} is $O(n\log{n})$
  and the size of the obtained graph $H'$ is linear in the size of $H$.
  Therefore, to finish the proof, we only need to solve our extended switch-on
  reachability problem on a graph $H$ having a simple recursive decomposition $\rtree(H)$.

  To this end, we extend the data structure of Section~\ref{sec:dag-trans} as follows.
  Recall that, by the simplicity of $\rtree(H)$, for each edge $e\in H$ there is a unique leaf
  subgraph $L\in\rtree(H)$ containing $e$.
  Also recall that the reachability information
  between the vertices of $L$ in $\swon{H}$
  can be recomputed in $O(1)$ time whenever the graph $\swon{L}$
  or the matrix $\extc(L)$ changes.
  In fact, for all of $O(1)$ edges $e\in L$
  and vertex pairs $u,v\in V(L)$, whether the path $\dirpath{u}{v}{}$
  exists in $\swon{H}-e$ can be computed based only on $\extc(L)$ and $\swon{L-e}$,
  when any of them changes, again in $O(1)$ time (see the proof of Lemma~\ref{lem:summatrices}).
  As the total update time remains $O(n\log^2{n}\log\log{n})$,
  the lemma follows.
\end{proof}

\begin{theorem}
Let $G$ be a planar digraph and let $n=|V(G)|$.
The maximal 2-edge-connected subgraphs of $G$ can be maintained
subject to edge deletions in $O(n\log^2{n}\log\log{n})$ total time.
\end{theorem}
\begin{proof}
It is known that the maximal 2-edge-connected subgraphs can be found
by repeatedly removing all the strong bridges of $G$ until none are left \cite{Georgiadis:2016}.
The maximal 2-edge-connected subgraphs of $G$ are defined as the strongly
connected components of the obtained subgraph.
Hence, we can combine Lemma~\ref{lem:dec-strong-bridge} with
our decremental strong connectivity algorithm
  to not only compute the maximal 2-edge-connected subgraphs (by repeatedly
  detecting and deleting the arising strong bridges)
but also to maintain them subject to edge deletions.
\end{proof}

%% file: incremental.tex
\subsection{Incremental Switch-On Transitive Closure}\label{sec:inctc}
\newcommand{\drg}{\ensuremath{\mathit{DRG}}}
\newcommand{\reachg}[2]{\ensuremath{\drg(#1,#2)}}

In this section we show how to solve the incremental transitive closure problem in planar graphs, under the assumption that the embedding of the final planar graph is given upfront.
We formulate this as a dynamic graph problem, in which the input is a directed graph and each edge can be either on or off.
Initially all edges are off, and an update operation may turn a single edge on.
Each query asks about the existence of a directed path that connects two vertices and consists of edges that are on.
The goal of this section is to prove the following theorem.

\begin{theorem}\label{thm:switch-on-tc}
There exists an algorithm that can maintain a planar digraph under switching edges on and can answer reachability queries.
Its total update time is $O(n \log^2 n \log \log n)$ and the query time is $O(\sqrt{n} \log n \log \log n)$.
\end{theorem}

An easy corollary is as follows.

\begin{theorem}
There exists an algorithm that can maintain a planar digraph under contracting edges and can answer reachability queries.
Its total update time is $O(n \log^2 n \log \log n)$ and the query time is $O(\sqrt{n} \log n \log \log n)$.
\end{theorem}

\begin{proof}
Let $G$ be the input graph.
We build a graph $G_{sw}$ which consists of $G$ together with all reverse edges of $G$.
The graph $G_{sw}$ is used as an input to the algorithm of Theorem~\ref{thm:switch-on-tc}.
Initially, we switch on all edges of $G$ in $G_{sw}$.

Whenever an edge of $G$ is contracted, we switch on its reverse edge in $G_{sw}$.
This implies that at any point, each vertex $u$ of $G$ corresponds to a strongly connected set of vertices of $\swon{G_{sw}}$.
We call this set a \emph{group} of $u$.
Clearly, the groups partition $V(G_{sw})$ and reachability between the groups in $\swon{G_{sw}}$ is the same as reachability in $G$.
Thus, in order to check whether there exists a path $\dirpath{u}{w}{G}$, it suffices to check whether there exists a path $\dirpath{u'}{w'}{\swon{G_{sw}}}$, where $u'$ and $w'$ are arbitrary vertices contained in the groups of $u$ and $w$, respectively.
\end{proof}

Let us now proceed with the proof of Theorem~\ref{thm:switch-on-tc}.
We first recall some definitions from Section~\ref{sec:dag-trans}.
For each $H$ of a \emph{simple} recursive decomposition $\rtree(G)$
we define $\intc(H)$ to be a binary
matrix with both rows and columns indexed with the vertices of $\bnd{H}$, such that $\intc(H)_{u,v}=1$ iff $u,v\in\bnd{H}$ and there exists a path $\dirpath{u}{v}{\swon{H}}$.
Note that we show that these matrices can be maintained under switch-ons in $O(n \log^2 n \log \log n)$ total time.
In addition, we can actually maintain the representation of these matrices used in the proof of Theorem~\ref{thm:trans} (see Lemma~\ref{lem:representation_a12}).

To complete the description of our algorithm, we show how to use these matrices to answer queries efficiently.
For a subgraph $H \in \rtree(G)$ we define the \emph{dense reachability graph}, denoted $\drg(H)$, to be a (possibly non-planar) graph on the vertices of $\bnd{H}$, such that $uv \in E(\drg(H))$ iff there is a path $\dirpath{u}{v}{\swon{H}}$.
Thus, $\drg(H)$ is fully described by $\intc(H)$, that is $\drg(H) = G(\intc(H))$ (recall Definition~\ref{def:mat-graph}).

In order to answer a query asking about the existence of a path $\dirpath{u}{w}{\swon{G}}$, we consider a graph $\reachg{u}{w}$, which is a union of dense reachability graphs (we give its full definition later).
This graph has $O(\sqrt{n})$ vertices in total, but can have as much as $\Theta(n)$ edges.
Then, we show how to search for a path in such a graph in time that is almost linear in the number of its \emph{vertices}.
This is somewhat similar to the FR-Dijkstra algorithm for finding shortest paths in dense distance graphs~\cite{FR06}, but due to lack of edge weights our algorithm is slightly faster.

Let us now describe how $\reachg{u}{w}$ is constructed.
Let $anc(H)$ be the set of ancestors of $H$ in $\rtree(G)$, including $H$, but \emph{excluding} the root of $\rtree(G)$.
Moreover, let $sib(H)$ be the sibling of $H$ in $\rtree(G)$.
Recall that $G(\intc(H))$ is the graph on $\bnd{H}$, whose each edge corresponds to a path in $\swon{H}$.

We can now define $\reachg{u}{w}$.
Let $H_u$ be any leaf node of $\rtree(G)$ that contains $u$ and $H_w$ be a leaf node containing $w$.
We have
\[
    \reachg{u}{w} := \swon{H_u} \cup \swon{H_w} \cup \bigcup_{H \in anc(H_u) \cup anc(H_w)} G(\intc(sib(H))).
\]

\begin{lemma}\label{lem:path_in_rtree}
Let $G$ be a planar digraph and $\rtree(G)$ be a simple recursive decomposition of $G$.
Let $u, w \in V(G)$ and $H_u$ and $H_w$ be leaf nodes of $\rtree(G)$ that contain $u$ and $w$.
For every path $P$ from $u$ to $w$, either $H_u = H_w$ and $P$ is contained in $H_u$, or $P$ contains a vertex $x$, such that $x \in \bnd{H'_u} \cap \bnd{H'_w}$, where $H'_u \in anc(H_u)$ and $H'_w \in anc(H_w)$.
\end{lemma}

\begin{proof}
Let $H$ be the lowest common ancestor of $H_u$ and $H_w$ in $\rtree(G)$.
There are three cases to consider.
If $H$ is a leaf node, then $H_u = H_w$.
Thus, every path from $u$ to $w$ is either fully contained in $H$ or leaves $H$ in which case it contains a vertex of $\bnd{H}$.
If $H$ is not a leaf node, but one of $u, w$ (say, $u$) is contained both in $\child_1(H)$ and $\child_2(H)$, then $u \in \bnd{\child_1(H)}\cap\bnd{\child_2(H)}$, so the lemma holds trivially.
Finally, if $H$ is not a leaf node, $u$ and $w$ are contained in distinct children of $H$.
Thus, the path has to leave the child node of $H$ containing $u$ (wlog. assume it to be $\child_1(H)$), so it contains a vertex of $\bnd{\child_1(H)}$.
By definition, such a vertex belongs to $\bnd{H}$ or to $\bnd{\child_1(H)}\cap\bnd{\child_2(H)}$.
Thus, the lemma follows.
\end{proof}

\begin{lemma}
There exists a path $\dirpath{u}{w}{\swon{G}}$ iff there exists a path $\dirpath{u}{w}{\reachg{u}{w}}$.
\end{lemma}
\begin{proof}
By the construction, it is clear that each path in $\reachg{u}{w}$ corresponds to a path that exists in~$\swon{G}$.
It remains to prove that if there is a path $\dirpath{u}{w}{\swon{G}}$, then there also exists a path $\dirpath{u}{w}{\reachg{u}{w}}$.
If $H_u=H_w$ and such a path is contained in $H_u$, this is immediate.

We show that $\reachg{u}{w}$ contains paths from $u$ to all vertices of $\bigcup_{H \in anc(H_u)} \bnd{H}$ that are reachable from $u$ in $\swon{G}$.
A symmetrical argument shows that $\reachg{u}{w}$ contains paths to $w$ from all vertices of $\bigcup_{H \in anc(H_w)} \bnd{H}$ that can reach $w$ in $\swon{G}$.
By Lemma~\ref{lem:path_in_rtree} this suffices to complete the proof.

Fix a vertex $x \in \bigcup_{H \in anc(H_u)} \bnd{H}$ that is reachable from $u$ in $\swon{G}$.
It suffices to show that $x$ is reachable from $u$ in $\swon{H_u} \cup \bigcup_{H \in anc(H_u)} G(\intc(sib(H)))$, which is a subgraph of $\reachg{u}{w}$.
Consider $H \in \rtree(G)$.
Observe that if we modify $\swon{G}$ by replacing $\swon{H}$ with $G(\intc(H))$, then $x$ is still reachable from $u$ in the obtained graph, as long as $u$ and $x$ are not removed, that is, $u, x \not\in V(H) \setminus \bnd{H}$.
To complete the proof, it suffices to observe that the graph $\swon{H_u} \cup \bigcup_{H \in anc(H_u)} G(\intc(sib(H)))$ is indeed obtained from $\swon{G}$ by performing such steps, as the graphs in $\{H_u\} \cup \{sib(H) \mid H \in anc(H_u)\}$ are edge-disjoint and each edge of $G$ is contained in exactly one of them; by the definition of boundaries, none of these replacements removes $u$ or $x$.
\end{proof}
In order to find a path in $\reachg{u}{w}$ we use an algorithm based~on breadth first search.
We implement it efficiently using ideas similar to those from the proof of Theorem~\ref{thm:trans}.
In the following pseudocode, $Out(a)$ denotes the set of heads of out-edges of $a$.
The set $Reachable$ is the set of vertices that have been deemed reachable from $u$.

\begin{algorithm}
\begin{algorithmic}[1]
\Require {A digraph $G$ and $u, w \in V(G)$.}
\Function{IsReachable}{$G, u, w$}

\State{$Q := $ empty queue}

\State{$Reachable := \{u\}$}
\State{$Q.\textsc{Enqueue}(u)$}

\While{$Q$ is not empty}
    \State $a := Q.\textsc{Dequeue}$
    \For {$x \in Out(a) \setminus Reachable$}
        \State{$Reachable := Reachable \cup \{x\}$}
        \State{$Q.\textsc{Enqueue}(x)$}
    \EndFor
\EndWhile
\State{\Return \textbf{true} iff $w \in Reachable$}
\EndFunction
\end{algorithmic}
\caption{\label{alg:bfs} BFS algorithm}
\end{algorithm}
\begin{lemma}\label{lem:bfs_rtime}
Algorithm~\ref{alg:bfs} is correct.
Its running time, excluding the time needed to compute
\linebreak
  $Out(a) \setminus Reachable$ is linear in $|V(G)|$.
\end{lemma}

\begin{proof}
The correctness is clear, as this is an ordinary BFS-algorithm.
The only difference is that the for-loop only considers vertices not yet in $Reachable$.
Clearly, the body of the for loop is executed at most once per each vertex of $G$ and each vertex of $G$ is added to $Q$ at most once.
Thus, the running time bound follows.
\end{proof}

\begin{lemma}\label{lem:reachg_matrices}
  The graphs comprising $\reachg{u}{w}$ can be decomposed into a collection
  \linebreak
  $\aparts_{u,w} = \{A^{S_1, T_1}, \ldots, A^{S_k, T_k}\}$
  of reachability matrices with property~(2) of Lemma~\ref{lem:monge_reach}, such that $\sum_i (|S_i|+|T_i|) = O(\sqrt{n} \log n)$
  and two constant-size arbitrary reachability matrices.
All these matrices are maintained by the algorithm that maintains the matrices $\intc(H)$.
\end{lemma}
\begin{proof}
  Recall that there are two types of graphs that comprise $\reachg{u}{w}$.
  First, there are two $O(1)$ size graphs $\swon{H_u}$ and $\swon{H_w}$ whose reachability
  matrices $\leaftc(H_u),\leaftc(H_w)$ are maintained explicitly by the data structure of Section~\ref{sec:dag-trans}.
  All remaining graphs that comprise $\reachg{u}{w}$ are of the form $G(\intc(H))$,
  where $H\in\rtree(G)$.
  By Lemmas~\ref{lem:monge_reach}~and~\ref{lem:boundary-curves}, each matrix
  $\intc(H)$
  can be decomposed into a family $\aparts_H$ of matrices with properties
  from the statement of Lemma~\ref{lem:monge_reach}
  such that each row of $\intc(H)$
  is a row of $O(\log{|\bnd{H}|})$ matrices of $\aparts_H$.
  Thus, if the level of $H$ in $\rtree(G)$ is $i$, the total number of rows in matrices of $\aparts_H$
  is $O(|\bnd{H}|\log |\bnd{H}|)=O(\sqrt{n}/c^i\log{n})$.
  The same bound holds for columns.

  We now define $\aparts_{u,w}=\bigcup_{H}\aparts_H$.
  Observe that for any level $i$ of $\rtree(G)$, there are at most two graphs
  $G(\intc(H))$ comprising $\reachg{u}{w}$, such that $H$ has level $i$.
  Thus, the total number of rows in the matrices of $\aparts_{u,w}$
  is $O\left(\sum_{i=0}^{O(\log{n})}\sqrt{n}/c^i\log{n}\right)=O\left(\sqrt{n}\log{n}\right),$
  since $c>1$ implies $\sum_{i=0}^\infty 1/c^i=O(1)$.
  To finish the proof, recall that the matrices of each $\aparts_H$ (and thus $\aparts_{u,w}$)
  are actually maintained explicitly by the data structure of Section~\ref{sec:dag-trans}.
\end{proof}

\begin{lemma}\label{lem:r-queries}
There exists an algorithm that determines whether there is a path $\dirpath{u}{w}{\reachg{u}{w}}$.
It runs in $O(\sqrt{n} \log n \log \log n)$ time.
\end{lemma}

\begin{proof}
We use the representation of $\reachg{u}{w}$ from Lemma~\ref{lem:reachg_matrices}.
  Denote the matrices of this representation by $\leaftc(H_u)$, $\leaftc(H_w)$ and
  $\aparts_{u,w} = \{A^{S_1, T_1}, \ldots, A^{S_k, T_k}\}$.

  For each $A^{S,T} \in \aparts_{u,w}$ the algorithm maintains the set of candidates $\actcols{A^{S,T}} \setminus Reachable$
as a van Emde Boas tree.
  For the two matrices $\leaftc(H_u), \leaftc(H_w)$, the set of candidates is maintained
  in an array.
  In the remaining part of the proof we neglect the matrices $\leaftc(H_u)$ and $\leaftc(H_w)$
  as they are constant-size and thus each operation on them can be performed in $O(1)$ time.

Whenever a vertex $x$ is added to the set $Reachable$, it is removed from each set of candidates containing it.
  In order to compute $Out(a) \setminus Reachable$, just as in the proof of Theorem~\ref{thm:trans}, for each $A^{S,T} \in \aparts_{u,w}$ such that $a \in S$ we compute the intersection of the set $\{b\in T:A^{S,T}_{a,b}=1\}$  and the set of reachability candidates.

The running time of computing $Out(a) \setminus Reachable$ and updating the set of reachability candidates can be thus bounded as follows.
  For each $A^{S,T} \in \aparts_{u,w}$, each element of $T$ is added and removed at most once from the set of reachability candidates.
This takes $O(\log \log n)$ time, as the set is represented with van Emde Boas tree.
By Lemma~\ref{lem:reachg_matrices}, this takes $O(\sqrt{n} \log n \log \log n)$ time, which dominates the running time of Algorithm~\ref{alg:bfs} (see Lemma~\ref{lem:bfs_rtime}).
\end{proof}
It remains to prove Theorem~\ref{thm:switch-on-tc}.
\begin{proof}[Proof of Theorem~\ref{thm:switch-on-tc}]
  We first apply Lemma~\ref{lem:simple-ext} to $G$ to obtain the graph $G'$
  such that $|E(G')|=O(n)$ and
  $G'$ has a simple recursive decomposition $\rtree(G')$.
  The computation of $\rtree(G'), E_0, E_1, \corr$, and $p$ (as defined
  in Lemma~\ref{lem:simple-ext}) takes $O(n\log{n})$ time.

  For each $H \in \rtree(G')$ we maintain $\intc(H)$ and its representation used
  in the proof of Theorem~\ref{thm:trans}
  (see Lemma~\ref{lem:representation_a12}) under switching edges on.
  Similarly as in the proof of Lemma~\ref{lem:simple-reduction}, all the
  edges of $E_0\subseteq E'$ are switched on in $G'$ during the initialization,
  whereas the edges $E'\setminus(E_0\cup E_1)$ are never switched on.
  To handle the switch-on of edge $e$ in $G$, we switch on the edge $p(e)$ in $G'$.

  Denote by $\swon{G}$ and $\swon{G'}$ the switched-on subgraphs of $G$ and $G'$, respectively.
  To decide whether there exists a path $u\to v$ in $\swon{G}$, by Corollary~\ref{col:subset-equiv}
  we can equivalently check whether there exists a path $u'\to v'$ in $\swon{G'}$, where
  $u'\in \corr(u)$ and $v'\in \corr(v)$.
  Answering the reachability queries in $G'$, by Lemma~\ref{lem:r-queries},
  takes $O(\sqrt{n} \log n \log \log n)$ time.
  The theorem follows.
\end{proof}

%% file: nonsimple.tex
\section{The General Case of Non-Simple and Overlapping Holes}\label{sec:nonsimple}

\newcommand{\ext}{\ensuremath{\mathcal{E}}}

Let $G=(V,E)$ be a planar digraph with $O(n)$ vertices and $O(n)$ edges.
This section is devoted to proving Lemma~\ref{lem:simple-ext}.
To this end, we transform our initial graph in a few steps.
Each step produces an extended graph $G'=(V',E')$
along with disjoint subsets $E_0,E_1\subseteq E'$ such
that:
\begin{itemize}
\item there is a 1-to-1 mapping $\corr$ between the vertices of $V$ and the SCCs
    of $(V',E_0)$, and $(V',E_0)$ has no inter-SCC edges.
\item there is a 1-to-1 mapping $p$ between $E$ and $E_1$ and
  for corresponding edges $uv=e\in E$ and $u'v'=p(e)\in E_1$,
    the SCC of $u'$ ($v'$, resp.)
    in $(V',E_0)$ corresponds to the vertex $u$ ($v$ resp.)
    of~$G$, i.e., $u'\in\corr(u)$ and $v'\in\corr(v)$.
\item The graph $G'$ is larger than $G$ by a constant factor.
\end{itemize}
Roughly speaking, each step will alter the graph $G$ obtained
in the previous step by expanding its vertices into strongly
connected components of $(V',E_0)$, mapping the edges $E$
to $E_1$ and adding some edges
that will belong to $E'\setminus (E_0\cup E_1)$
and guarantee some structural property that we want.
For brevity, we sometimes only discuss the correspondences between the SCCs of $(V',E_0)$
and vertices $V$ and edges
mapped to the initial edge set $E$.
We may skip the formal definitions of the mappings $\corr$ and $p$.

It can be easily seen that after $O(1)$ such steps the obtained
graph is still larger than $G$ only by a constant factor.
In each step we obtain a graph with more useful structural properties
and our goal is to eventually
obtain a graph accompanied with a simple recursive
decomposition.

Assume some first step transforms $G=(V,E)$ into $G'=(V',E')$ where
$E_0',E_1'\subseteq E'$ and some second step transforms $G'$ into $G''=(V'',E'')$,
where $E_0'',E_1''\subseteq E''$.
We explain how the natural composition of such steps works, i.e., how the mappings between $G$ and $G''$,
along with the sets $E_0,E_1$, should look like after applying
these steps in sequence.
Let $p'$ be the 1-to-1 mapping from $E$ to $E_1'$ and let $p''$ be the 1-to-1 mapping
from $E'$ to $E_1''$.
We define $E_1=p''(p'(E))$ and $p=p''\circ p'$.
Moreover, $E_0=E_0''\cup p''(E_0')$.
Observe that $E_0,E_1\subseteq E''$ and $E_0\cap E_1=\emptyset$.
Note that $(V'',E_0'')$ has no inter-SCC edges and there is a mapping $\corr''$ between
$V'$ and the SCCs of $(V'',E_0'')$.
Moreover, the insertion of edges $p''(E_0')$ merges some of the SCCs of
$(V'',E_0'')$ so that there are still no inter-SCC edges and there is still a natural 1-to-1 mapping
between the SCCs of $(V'',E_0)$ and $V$.
\subsection{Bounded-degree, Triangulated and Without Loops}
Here, by \emph{triangulated} we mean that each face has size exactly $3$.
We use the following standard fact.
\begin{fact}\label{fac:biconn}
A plane embedded graph $G$ has only simple faces if and only if $G$ is biconnected.
\end{fact}
The first step is to make $G$ connected by adding a minimal number of edges:
for the obtained $G'=(V',E')$,
where $V=V'$ and $E\subseteq E'$, we set $E_0=\emptyset$ and $E_1=E$.
Clearly, $G'$ is also planar.
In the following, we assume $G$ is connected.

We next fix some embedding of $G$ in $O(n)$ time \cite{Hopcroft:74}.
Define an \emph{edge ring} of $v$ to be the sequence of edges $e_1,\ldots,e_k$
incident to $v$ and ordered clockwise as implied by the embedding.
Note that if some $e_i$ is a loop, it appears in the edge ring of $v$ twice.

We compute $G'$ based on $G=(V,E)$ as follows.
Let $v$ be a vertex of $G$ with an edge ring $e^v_1,\ldots,e^v_k$.
Define $C(v)$ to be a directed cycle $v_1v_2\ldots v_{3k}$,
i.e., $C(v)$ contains directed edges $v_iv_{i+1}$ for each $i=1,\ldots,3k$,
where $v_{3k+1}=v_1$.
Define $V'=\bigcup_{v\in V} V(C(v))$ and $E_0=\bigcup_{v\in V} E(C(v))$.
We also let $\corr(v)=V(C(v))$ -- clearly for each $C(v)$, $V(C(v))$ is strongly connected.

For any $uv=e\in E$, suppose the tail of $e$ appears as the $i$-th edge in the edge ring of $u$
and its head appears as the $j$-th edge in the edge ring of $v$.
This is also well-defined in the case when $e$ is a loop.
We include the edge $u_{3i}v_{3j}$ in $E_1$.
Clearly, the edge $u_{3i}v_{3j}$ goes from a vertex in a SCC of $(V',E_0)$
corresponding to $u$ to a vertex in a SCC corresponding to $v$.

Note that $H=(V',E_0\cup E_1)$ is connected and let it naturally inherit the embedding from $G$.
Observe that each face of $H$ either constitutes the interior of some cycle $C(v)$ and has
size at least $3$, or is an ``expanded'' version of one of the original faces
of $G$ and its bounding walk has length at least $4$.
Hence, we conclude that each face of $H$ is of size at least $3$.

Additionally, note that the underlying undirected graph of $H$ is simple
and has maximum degree at most $3$.
Indeed, each original edge incidence has its own vertex on a cycle $C(v)$;
even an original loop joins distinct, non-neighboring vertices of that cycle.
We apply the simplicity-preserving triangulation algorithm
of \cite[Theorem~6.2.1]{Kant93} to this underlying graph.
It adds no vertices, runs in linear time, and produces a simple triangulated
plane graph of maximum degree at most $\lceil3\cdot3/2\rceil+11=16$.
The algorithm may choose a new embedding, which is immaterial to the
correspondences $\corr$ and $p$.
Keep the directions of the old edges and orient the added edges arbitrarily,
including them in $E'\setminus(E_0\cup E_1)$.
Thus $G'$ is triangulated, has constant degree, and its underlying graph is simple.

Clearly, the described step can be implemented in $O(n)$ time and
produces a graph $G'$ with $|V'|+|E'|=O(n)$:
the total length of the edge rings is twice the number of edges,
and making the initial graph connected adds at most $|V|-1$ edges.
\subsection{Admitting a Simple Recursive Decomposition}
In this section we show how to transform a triangulated $G$ of bounded
degree into $G'$ such that we can expose a simple recursive decomposition of $G'$.

We first show how we build a (not necessarily simple) recursive decomposition of a planar digraph.
\begin{lemma}\label{lem:klein}
  Let $G=(V,E)$ be a simple, connected, triangulated plane graph and let $n=|V|$.
  A recursive decomposition tree $\rtree(G)$ of $G$ can be computed in $O(n\log{n})$ time.
\end{lemma}
\begin{proof}
We use the cycle-separator recursion of \cite{Klein:13}, with the
qualifications below, and prove the quantitative bounds directly.
Stop when a piece has at most $n_0$ vertices, for a sufficiently large
constant $n_0>9$.
To split a piece $H$, put a temporary vertex $v_h$ in each hole $h$ and
join it to every corner of $h$. Repeated occurrences of a vertex on a
hole give distinct spokes. The resulting plane multigraph $H^\Delta$
is loopless, triangulated, and biconnected, as in \cite[Lemma~3]{Klein:13}.
Apply Miller's cycle-separator theorem \cite{Miller:86} with nonnegative
vertex weights to $H^\Delta$.
The vertices strictly on either side of the separator have at most
$2/3$ of the total weight. At successive levels, use the following weights:
\begin{itemize}
\item At levels congruent to $0$ modulo $3$, give unit weight to each vertex of $H$.
\item At levels congruent to $1$ modulo $3$, give unit weight to each vertex of
  $\bnd H$ if $|\bnd H|>9$; otherwise give unit weight to each vertex of $H$.
\item At levels congruent to $2$ modulo $3$, give unit weight to each hole
  vertex $v_h$ if $H$ has a hole; otherwise give unit weight to each vertex of $H$.
\end{itemize}
All other weights are zero.
If the separator is a two-edge cycle $v_hvv_h$, replace one spoke by
$v_hw,wv$, along an incident artificial triangular face $v_hvw$.
Here $w\neq v$ because $H$ is loopless. This yields a simple triangle,
adds one separator vertex, and preserves balance: it sweeps an empty
artificial face and moves only $w$ onto the separator.
Suppress the hole vertices on the resulting curve, obtaining a Jordan
curve $C_H$ whose intersections with $H$ are vertices and edges.
A passage through a hole vertex becomes a \emph{hole-part}; the remaining
parts follow edges of $H$.

The two children are the subgraphs weakly on either side of $C_H$.
They are connected: replacing each hole-part by the hole-boundary walk
on the corresponding side gives a closed walk through the separator
vertices, along which paths can be rerouted to stay in that child.
The same description shows that a child has at most one new hole;
the other holes are unchanged holes of $H$ on that side.
A separator bounding a non-hole triangular face is excluded by our
weight choices. Such a triangle cannot balance more than nine unit-weight
vertices, and all hole vertices lie strictly on one side of it.
Thus the face containing the discarded side really is a hole of the child,
and contains all its separator vertices.
Consequently, all boundary vertices lie on holes, by induction.

Each split increases the number of holes by at most one. In a
hole-balancing round, a child has at most $2/3$ of the old holes plus one,
so the three-round schedule keeps the number of holes bounded by a constant.
Miller's separator has $O(\sqrt{|V(H^\Delta)|})=O(\sqrt{|V(H)|})$ vertices;
the normalization above preserves this bound.

We give a self-contained analysis of the quantitative bounds.
Increase a constant $a$ so that each separator of an $m$-vertex piece has
at most $a\sqrt m$ vertices. If that piece has $b$ boundary vertices,
each child has at most $m$ vertices and $b+a\sqrt m$ boundary vertices.
In a vertex-balancing round its vertex count is at most
$\frac23m+a\sqrt m$, and in a boundary-balancing round its boundary count
is at most $\frac23b+a\sqrt m$.
The latter inequality also holds when $b\leq9$ and vertex weights are used
instead, by increasing $a$ to absorb the additive constant $b/3\leq3$.

Choose the stopping threshold $n_0$ sufficiently large that
$\frac23m+a\sqrt m\leq\rho m$ for $m>n_0$, where $\rho=5/6$.
Every three consecutive rounds include a vertex-balancing round.
Consequently, a descendant $i$ levels below an $m$-vertex piece has
at most $m\rho^{\lfloor i/3\rfloor}$ vertices.
This proves that the height is $O(\log n)$; in particular, a subtree
whose root has a bounded number of vertices has bounded size.
Along any root-to-leaf path, write $b_i$ for the boundary count at level $i$.
The boundary-balancing round in each group of three gives
\[
 b_{3j+3}\leq\tfrac23 b_{3j}+3a\sqrt n\,\rho^{j/2}.
\]
Since $b_0=0$ and $2/3<\sqrt\rho$, summing this geometric recurrence gives
$b_{3j}=O(\sqrt n\,\rho^{j/2})$.
The intermediate levels satisfy the same bound up to a constant factor.
Thus $|\bnd H|=O(\sqrt n/c^{\mathrm{level}(H)})$ for
$c=\rho^{-1/6}>1$.

To bound the total number of pieces and the sum of squared boundary sizes,
fix an integer $K$ divisible by three with $(2/3)^{K/3}\leq1/3$,
and set $Q=2^K$.
Consider the subtree of a piece $H$ with $m$ vertices and $b$ boundary
vertices, truncated after $K$ levels or at earlier leaves.
Let $H_1,\ldots,H_q$, where $q\leq Q$, be its terminal pieces, and put
$m_i'=|V(H_i)|$ and $b_i'=|\bnd{H_i}|$.
All intermediate pieces have at most $m$ vertices and
$b+Ka\sqrt m$ boundary vertices. At level $K$ the balancing inequalities give
\[
 m_i'\leq m/3+Ka\sqrt m,\qquad
 b_i'\leq b/3+Ka\sqrt m.
\]
For each $i$, let $X_i$ be the vertices of $H_i$ on separators of its
ancestors in this truncated tree. Then $|X_i|\leq Ka\sqrt m$.
The sets $V(H_i)\setminus X_i$ are pairwise disjoint, and the sets
$\bnd{H_i}\setminus X_i$ are pairwise disjoint subsets of $\bnd H$.
Indeed, intersections of different terminal pieces lie on an intervening
separator, and every boundary vertex not inherited from $H$ lies on one
of these separators.
Set $s=\lceil Ka\sqrt m\rceil$.
Remove from $V(H_i)$ a set containing $X_i$ of size $\min(m_i',s)$,
and from $\bnd{H_i}$ a set containing $X_i\cap\bnd{H_i}$
of size $\min(b_i',s)$. Denote the remaining cardinalities by $m_i,b_i$.
Set $d=Ka+1$. For a sufficiently large constant $m_*$, whenever
$m\geq m_*$ we have $s\geq n_0$ and $m/3+d\sqrt m<m/2$.
If $H_i$ is a leaf reached before depth $K$, then
$b_i'\leq m_i'\leq n_0\leq s$, so the above removals leave $m_i=b_i=0$.
We obtain
\begin{gather*}
 \sum_i m_i\leq m,\quad 0\leq m_i\leq m/3,\qquad
 \sum_i b_i\leq b,\quad 0\leq b_i\leq b/3,\\
 m_i\leq m_i'\leq m_i+d\sqrt m<m/2,\qquad
 b_i\leq b_i'\leq b_i+d\sqrt m.
\end{gather*}

Let $T(H)=\sum_J(|V(J)|+|\bnd J|^2)$, where the sum is over the subtree
rooted at $H$. There are at most $Q-1$ intermediate pieces in the truncated
tree, each contributing at most $m+(b+Ka\sqrt m)^2$. Their total contribution
is therefore at most $C(m+b^2)$, for a constant $C$ depending only on $K$
and $a$. Thus
\[
 T(H)\leq C(m+b^2)+\sum_i T(H_i).
\]
We prove inductively that $T(H)\leq A m\log_2 m+B b^2$.
Enlarge $m_*$ so that $Qd\sqrt m\log_2 m\leq m/2$ for $m\geq m_*$.
The bounds above imply
\begin{align*}
 \sum_i m_i'\log_2 m_i'
   &\leq (m+Qd\sqrt m)\log_2(m/2)
    \leq m\log_2 m-m/2,\\
 \sum_i(b_i')^2
   &\leq 2\sum_i b_i^2+2Qd^2m
    \leq 2b^2/3+2Qd^2m.
\end{align*}
Here $\sum_i b_i^2\leq (\max_i b_i)\sum_i b_i\leq b^2/3$.
The induction therefore closes by taking $B\geq3C$ and
$A\geq2C+4BQd^2$, and increasing $A$ to cover $m<m_*$.
Those cases have bounded subtree size by the height bound already proved.
At the root $b=0$, giving both
$\sum_H|V(H)|=O(n\log n)$ and $\sum_H|\bnd H|^2=O(n\log n)$.

For completeness, the same truncated tree also bounds the number of leaves.
Let $U(H)$ be the number of pieces in the subtree rooted at $H$.
Then $U(H)\leq Q+\sum_i U(H_i)$.
Put $\phi(m)=\max(m-m^{0.6},1)$ and prove $U(H)\leq A\phi(m)$,
with a possibly different constant $A\geq1$.
Let $S$ contain the indices for which $\phi(m_i')=m_i'-(m_i')^{0.6}$,
and write $t=\sum_{i\in S}m_i$.
The inductive bound and $(m_i')^{0.6}\geq m_i^{0.6}$ give
\[
 U(H)/A\leq 2Q+Qd\sqrt m+t-\sum_{i\in S}m_i^{0.6}.
\]
If $t\leq2m/3$, the right-hand side is at most $m-m^{0.6}$
for sufficiently large $m$.
Otherwise, since $m_i\leq m/3$,
\[
 \sum_{i\in S}m_i^{0.6}\geq t(m/3)^{-0.4}
   \geq2(m/3)^{0.6}.
\]
As $2/3^{0.6}>1$, this again makes the right-hand side at most
$m-m^{0.6}$ for sufficiently large $m$.
Increase the fixed cutoff and then $A$ to cover the remaining bounded-size
subtrees. This proves $U(H)=O(m)$ and, in particular, $O(n)$ leaves.
Each leaf has at most $n_0$ vertices and hence $O(1)$ edges, because it is
a subgraph of the simple planar graph $G$.
Finally, planarity and $\sum_H|V(H)|=O(n\log n)$ imply that
the linear work at each separator step totals $O(n\log n)$.
\end{proof}

We record the structural properties of the separators in the preceding
construction. Corners are occurrences in facial boundary walks; a hole need
not have a simple boundary.

\begin{lemma}\label{lem:cycle-separator}
  Let $G$ be a simple, connected, triangulated plane graph.
  Let $H$ be a non-leaf subgraph of $\rtree(G)$, where
  $\rtree(G)$ is a decomposition produced by Lemma~\ref{lem:klein}.
  There exists a Jordan curve $C_H$ (called a \emph{cycle separator})
  going through distinct vertices $v_1,\ldots,v_k\in V(H)$ (where $v_{k+1}=v_1$) such that:
  \begin{itemize}
    \item The children of $H$ are exactly the subgraphs of $H$ contained
      respectively weakly-inside or weakly-outside $C_H$.
    \item For each $i\in\{1,\ldots,k\}$, the part of $C_H$ going from $v_i$ to $v_{i+1}$ either
      follows some (undirected) edge $v_iv_{i+1}$ of $H$ (we call it an \emph{edge-part})
      or otherwise its interior is strictly contained inside some hole $h$ of $H$
      (we call it a \emph{hole-part}). The endpoints are distinct vertices; their corners may be neighboring on the boundary of $h$.
    \item For each hole $h$ of $H$, there exists at most one hole-part of $C_H$
      inside $h$.
    \item For $i=1,2$, all vertices of $C_H$ belong to $\bnd{\child_i(H)}$
      and lie on a single hole $h^*$ of $\child_i(H)$.
      All holes $h\neq h^*$ of $\child_i(H)$ are also holes of $H$ and are located
      (weakly) on the same side of $C_H$ in $H$ as $\child_i(H)$.
    \item The set $E(\child_1(H))\cap E(\child_2(H))$ contains exactly
      the edges corresponding to the edge-parts of $C_H$,
      whereas $V(\child_1(H))\cap V(\child_2(H))=\{v_1,\ldots,v_k\}$.
  \end{itemize}
  All curves $C_H$ can be output as part of the $O(n\log n)$-time
  construction of Lemma~\ref{lem:klein}.
\end{lemma}
\begin{proof}
Each temporary hole vertex occurs at most once on the simple separator,
so there is at most one hole-part per hole. The two-edge-cycle
normalization ensures that its endpoints are distinct.
Taking the weak sides gives exactly the stated edge and vertex intersections.
As shown in the preceding proof, the face containing the discarded side
is a hole $h^*$ and its bounding walk contains all separator vertices.
Any other hole is an unchanged hole of the parent on the retained side.
The separator curves are available when the pieces are split.
\end{proof}

Let $G$ be a plane embedded digraph of bounded degree whose underlying
undirected graph is simple and triangulated.
Let $\rtree(G)$ be a recursive decomposition built as in Lemma~\ref{lem:klein}
(and possessing the properties summarized by Lemma~\ref{lem:cycle-separator}).
We now show how to use $\rtree(G)$ to build a \emph{simple} (see Definition~\ref{def:simple-decomp})
recursive decomposition $\rtree(G')$ of a transformed graph $G'=(V',E')$.

The decomposition $\rtree(G')$ has the same shape as $\rtree(G)$ in the following sense.
\begin{itemize}
  \item There is a $1$-to-$1$ correspondence $\ext:\rtree(G)\to\rtree(G')$ between
    the subgraphs of $\rtree(G)$ and the subgraphs of $\rtree(G')$.
  \item For the root subgraph $G\in\rtree(G)$, $\ext(G)=G'$ is the root of $\rtree(G')$.
    Moreover, $H\in\rtree(G)$ has children iff $\ext(H)$ has children in $\rtree(G')$ and $\ext(\child_i(H))=\child_i(\ext(H))$ for $i=1,2$.
\end{itemize}
The graph $G'$ has the following additional properties which
make the transformation $G\to G'$ comply with the requirements
posed on a transformation step at the beginning of this section.
\begin{enumerate}
  \item Each vertex $v\in V'$ is mapped to a vertex $\orig(v)\in V(G)$.
  \item Denote by $E_0$ the subset of edges of $G'$ connecting vertices
    $u,v$ of $G'$ such that $\orig(u)=\orig(v)$.
    There is a 1-to-1 mapping $\corr$ between the vertices of $G$ and
    the strongly connected components of $(V',E_0)$.
  \item For any $uv=e\in E$ there is an edge $u'v'=p(e)\in E'\setminus E_0$ such that
    $\orig(u')=u$ and $\orig(v')=v$ and $p$ is injective.
\end{enumerate}

\subsubsection{Overview of the Construction}
For each $H\in\rtree(G)$, $H$ is conceptually extended, so that:
\begin{itemize}
  \item Each vertex $v\in V(H)$ is expanded to a directed cycle of sub-vertices of $\ext(H)$.
    The cycle is further split into smaller faces, corresponding
    to the expansions of $v$ in the descendants of $\ext(H)$ in $\rtree(G')$.
  \item Each edge $e\in E(H)$ is expanded to a \emph{undirected cycle} (i.e., a cycle,
    but when ignoring the direction of edges) in $\ext(H)$.
    The cycle is further split into smaller faces, corresponding
    to the expansions of $e$ in the descendants of $\ext(H)$.
  \item The subgraphs $\child_1(\ext(H))$ and $\child_2(\ext(H))$ are
    obtained from $\ext(H)$ using the same (in a sense) cycle separator
    as used to obtain $\child_1(H)$ and $\child_2(H)$ of $H$.
  \item There is a $1$-to-$1$ correspondence between the faces of $H$
    and so-called \emph{original} faces of $\ext(H)$.
    In particular, there is a $1$-to-$1$ correspondence
    between the holes of $H$ and $\ext(H)$.
  \item By introducing \emph{vertex faces} and \emph{edge faces}, we separate the original faces 
    (and thus the original \emph{holes} as well) of $\ext(H)$,
    so that they do not share vertices.
  \item The size of each original hole of $\ext(H)$ is larger by a constant factor than the
    size of the corresponding hole of $H$.
  \item The leaf subgraphs of $\ext(H)$ are of constant size as well.
\end{itemize}
Figure~\ref{fig:extended} shows how a subgraph $H\in\rtree(G)$ can be possibly expanded.
Figure~\ref{fig:extended-split} shows how obtaining the children of $H\in\rtree(G)$
with a cycle separator $C_H$ translates to obtaining the children of $\ext(H)$.

The remaining part of this section is devoted to providing the technical
details of the construction sketched above.

\subsubsection{Details of the Construction}
Let $H\in\rtree(G)$.
We first establish some structural requirements posed on the subgraphs of $\rtree(G')$.
Each vertex $v\in V(H)$ corresponds to a
simple cycle $Q^H(v)=v_1,\ldots,v_{k_v}$ of vertices of $V(\ext(H))$ (ordered clockwise).
For each of these $v_i$, we set $\orig(v_i)=v$.
Moreover, assuming $v_{k_v+1}=v_1$, we have a directed
edge $v_iv_{i+1}\in E(\ext(H))$ for each $i\in [1,k_v]$.
These edges are called \emph{border vertex edges} of $v$ (with respect to $H$).

Let $e_1,\ldots,e_d$ be the edge ring of $v$ in $H$.
For $i=1,\ldots,d$, there is a contiguous subsequence $Q^H_{e_i}(v)$ of
vertices of $Q^H(v)$.
The fragments $Q^H_{e_1}(v),\ldots,Q^H_{e_d}(v)$ are disjoint and put
on $Q^H(v)$ in the same order as the edges $e_1,\ldots,e_d$ lie in the edge-ring
of $v$.
However, there might be some other vertices in $Q^H(v)$ put between
any two neighboring fragments $Q^H_{e_i}(v)$ and $Q^H_{e_{i+1}}(v)$.
We call $\bigcup_{i=1}^d Q^H_{e_i}(v)$ the \emph{edge vertices} of $Q^H(v)$.
An edge vertex that is either first or last in some $Q^H_{e_i}(v)$
is called a \emph{border edge vertex} of $v$.
The remaining vertices $X^H(v)=Q^H(v)\setminus \bigcup_{i=1}^d Q^H_{e_i}(v)$ are called the \emph{extra vertices}
of $Q^H(v)$.
If $d=1$, then we require $X^H(v)\neq\emptyset$, so that the clockwise
order of $Q^H_{e_1}(v)$ is well-defined.
The extra vertices of $Q^H(v)$ can only be adjacent (in undirected sense) to other
vertices of $Q^H(v)$.

There may be some additional edges between the non-adjacent (but distinct) vertices
of $Q^H(v)$.
All such edges are always embedded inside the cycle $Q^H(v)$ and do not cross.
They partition the inside of all these cycles into faces, which
are collectively called the \emph{vertex faces} of $\ext(H)$.
All edges connecting pairs of vertices of $Q^H(v)$ are called \emph{vertex edges}
of $v$ (with respect to $H$).

Let $uv=e\in E(H)$. Then, the fragments $Q^H_e(u)=u_1,\ldots,u_l$ and $Q^H_e(v)=v_1,\ldots,v_l$ have
equal lengths and $l\geq 2$. Recall that $u\neq v$, as we assume that $G$ has no loops.

We have two edges $u_1v_l$ and $u_lv_1$, which are called the \emph{border edge edges} of $e$.
Additionally,
for each $i=2,\ldots,l-1$, there are two parallel edges $u_iv_{l-i+1}\in E(\ext(H))$.
These edges are called the \emph{internal edge edges} of $e$ (with respect to $H$).
All edge edges of $e$ have the direction of $e$, from $Q^H_e(u)$ to $Q^H_e(v)$.
The $l-1$ rectangular faces $u_iv_{l-i+1}v_{l-i}u_{i+1}$ (for $i=1,\ldots,l-1$)
and $l-2$ faces $u_iv_{l-i+1}$ of length $2$ are collectively
called the \emph{edge faces} of $e$ (with respect to $H$).
Consult Figure~\ref{fig:extended} for better understanding of how edge edges are constructed.

All the remaining faces of $\ext(H)$ that are neither vertex faces nor edge faces
are the \emph{original faces} of $\ext(H)$.
This name is justified by the fact that there is a 1-to-1 correspondence between
such faces and faces of $H$.

\begin{figure}
  \centering
  \includegraphics{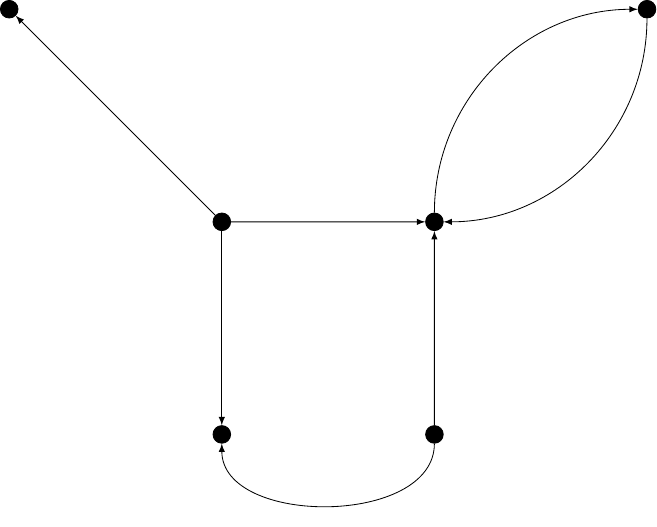}
  \includegraphics{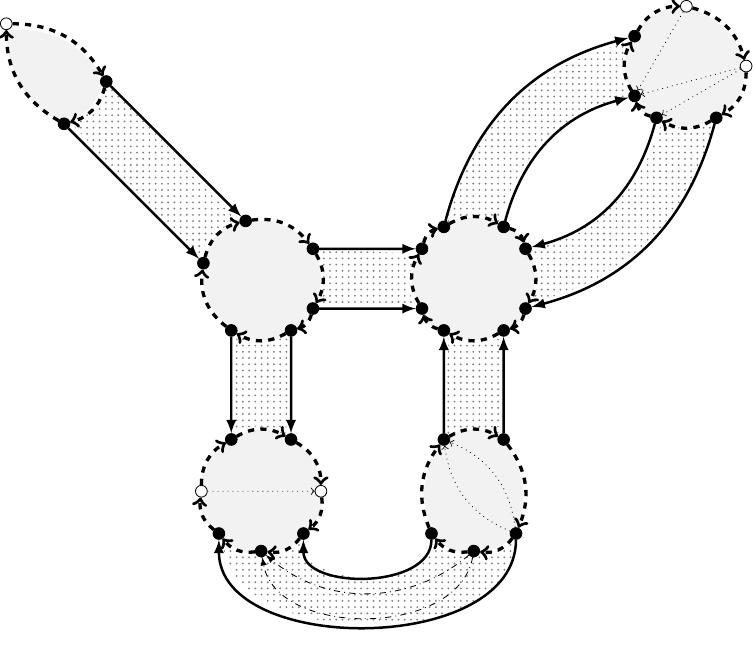}
\caption{A directed embedded plane graph $H$ (top) and an example extended graph $\ext(H)$ (bottom).
The non-filled vertices are the extra vertices.
Solid edges in the bottom picture are the border edge edges of $\ext(H)$, thick dashed edges
are the border vertex edges, thin dashed edge is a non-border edge edge, whereas
the dotted edges are non-border vertex edges.
The white faces are the original faces, gray faces are the vertex faces, while the dotted
faces are edge faces.
}
\label{fig:extended}
\end{figure}

\begin{lemma}\label{lem:hole-disjoint}
No two original faces of $\ext(H)$ share vertices.
\end{lemma}
\begin{proof}
An internal edge vertex is incident only to vertex faces and edge faces.
A border edge vertex is incident to two faces outside its vertex cycle,
one of which is an edge face. An extra vertex is incident to only one face
outside its vertex cycle. All faces inside the vertex cycles are vertex faces.
Thus every vertex is incident to at most one original face, proving the claim.
\end{proof}
\begin{lemma}\label{lem:hole-simple}
All faces of $\ext(H)$ are simple.
\end{lemma}
\begin{proof}
  We prove that $\ext(H)$ is biconnected, or equivalently,
  that any two distinct edges $e,e'\in E(\ext(H))$ (ignoring their directions)
  lie on a simple cycle (in an undirected sense).
  By Fact~\ref{fac:biconn}, this will imply that $\ext(H)$ has only simple faces.
  Recall that biconnectivity is an equivalence relation on edges.

  If $e,e'$ are edge edges of the same original edge $uv=e^*\in E(H)$,
  join their endpoints by the corresponding paths in $Q^H_{e^*}(u)$
  and $Q^H_{e^*}(v)$. Together with $e,e'$ these paths form a simple
  cycle (a two-edge cycle when $e,e'$ are parallel).

  The border vertex edges of some $v$ all lie on a simple cycle by their
  definition.
  Moreover, each non-border vertex edge of $v$ is embedded inside this cycle
  and thus is biconnected with all the border vertex edges.
  This establishes the biconnectivity of vertex edges of a single vertex $v$.

  Note that the border edge edges of any $uv=e\in E(H)$ lie on a simple cycle
  that includes both the border vertex edges of $u$ and $v$.
  Thus, all edge edges of $e$ are biconnected with both the vertex edges of $u$
  and vertex edges of $v$.

  Observe that as $H$ is connected, the above is enough to conclude
  that $\ext(H)$ is biconnected.
\end{proof}
It is clear that $E_0=\{u'v'\in E':\orig(u')=\orig(v')\}$ contains only the
vertex edges of $E'$ and thus there is a 1-to-1 correspondence between
the vertices of $V$ and the SCCs of $(V',E_0)$.
For $e\in E$, we include in $E_1$ any edge edge $e'$ of $e$ and set $p(e)=e'$.

We now have to guarantee that $\rtree(G')$ has all properties of a recursive
decomposition, and to this end we pose additional requirements on $\ext(H)$.
Let $D$ be a constant bounding the degree of $G$.
\begin{enumerate}[label=B.\arabic*]
  \item If $v\in V(H)$ and $v\notin\bnd{H}$, then $\bnd{\ext(H)}\cap Q^H(v)=\emptyset$.\label{prop:b1}
  \item If $v\in\bnd{\ext(H)}$ and $v$ is an edge vertex of $Q^H(\orig(v))$, then
    $v$ is a border edge vertex of $Q^H(\orig(v))$.\label{prop:b2}
  \item For $v\in V(H)$, the set $X^H(v)$ has size at most $\deg_G(v)-\deg_H(v)+1$.\label{prop:b3}
  \item The holes of $\ext(H)$ are exactly the original faces of $\ext(H)$
    corresponding to the holes of $H$.\label{prop:b4}
  \item If $H$ is a leaf subgraph of $\rtree(G)$, then every edge vertex of $\ext(H)$
    is a border edge vertex and $\ext(H)$ has no non-border vertex edges.\label{prop:b5}
\end{enumerate}
\begin{lemma}\label{lem:border-size}
  The above requirements imply that:
  \begin{enumerate}
    \item\label{ei:sumsz} $\sum_{H\in\rtree(G)}|\bnd{\ext(H)}|^2=O(n\log{n})$.
    \item\label{ei:geom} For any $H\in\rtree(G)$ of the $i$-th level of $\rtree(G)$, $|\bnd{\ext(H)}|=O(\sqrt{n}/c^i)$,
      where $c>1$ is a constant.
    \item Each leaf subgraph of $\ext\rtree(G)$ has $O(1)$ edges.
  \end{enumerate}
\end{lemma}
\begin{proof}
  Let $H\in\rtree(G)$.
  By~\ref{prop:b1}, only vertices of $\bnd H$ contribute to $\bnd{\ext(H)}$.
  By~\ref{prop:b2} and~\ref{prop:b3}, each such $v$ contributes at most
  $2\deg_H(v)$ edge vertices and $\deg_G(v)-\deg_H(v)+1$ extra vertices.
  Since $\deg_H(v)\leq\deg_G(v)\leq D$, we have
  $|\bnd{\ext(H)}|\leq(2D+1)|\bnd H|=O(|\bnd H|)$.
  This implies both properties~\ref{ei:sumsz}~and~\ref{ei:geom}.

  Let $H$ be a leaf subgraph of $\rtree(G)$. Recall that $H$ has $O(1)$ edges.
  By~\ref{prop:b3}, \ref{prop:b5} and
  the structural requirements posed on $\rtree(G')$,
  for each $v\in V(H)$, $Q^H(v)$ has size at most $2\deg_H(v)+|X^H(v)|=O(1)$ and
  there are only $\sum_{v\in V(H)}|Q^H(v)|=O(1)$ vertex edges.
  Thus, $\ext(H)$ has at most $O(1)$ edges.
\end{proof}

Now we show how to actually build the subgraphs of $\rtree(G')$ so that all
the described requirements are met and that $\rtree(G')$ has all the needed
properties of a recursive decomposition.
This, combined with Lemmas~\ref{lem:hole-disjoint}, \ref{lem:hole-simple}~and~\ref{lem:border-size}
will finish the proof.

We will start with the root of $\rtree(G)$ and gradually move down the tree $\rtree(G)$,
introducing new edge vertices and edge edges of $G'$ once they are needed.

Each subgraph $\ext(H)$ is \emph{created} once. The boundary vertices of $\ext(H)$
do not change after it is created, but the graph $\ext(H)$ itself might
grow in a controlled way.
We start by creating the root $\ext(G)$ as a unique graph such that:
\begin{itemize}
\item $|Q^G_e(v)|=2$ for each $e\in E(G)$ such that $v$ is an endpoint of $e$.
\item $X^G(v)=\emptyset$ if $\deg_G(v)>1$ and $|X^G(v)|=1$ otherwise.
\end{itemize}
We set $\bnd{\ext(G)}=\emptyset$.
Clearly, this satisfies the requirements~\ref{prop:b1}-\ref{prop:b4}
as $\bnd{G}=\emptyset$ and there are no holes in the root subgraph of $\rtree(G)$.
Moreover, we can consider $\ext(G)$ as a leaf of $\rtree(G')$ at this
point, as its children have not been created yet.
The property~\ref{prop:b5} that characterized leaves is indeed satisfied
for $\ext(G)$ at this point.

Each subsequent step of the construction
will consist of taking a current leaf subgraph of $\rtree(G')$ and
creating its children.
We will proceed in such a way that the properties~\ref{prop:b1}-\ref{prop:b4}
will remain satisfied for all created subgraphs at all times.
On the other hand, the property~\ref{prop:b5} will cease to be satisfied
when subgraph stops to be a leaf of $\rtree(G')$.

Now let $H$ be a non-leaf subgraph of $\rtree(G)$ such that $\ext(H)$ has been created,
as opposed to its children.
Let $C_H=v_1,\ldots,v_k$ (where $v_{k+1}=v_1$) be the cycle separator
of Lemma~\ref{lem:cycle-separator}.

For each edge part $v_iv_{i+1}$ of $C_H$ corresponding to an edge $e\in E(H)$, we first create new edge vertices
$b_i,a_{i+1}$ by placing $b_i$ in $Q^H_e(v_i)$ between the only two (see~\ref{prop:b5})
vertices of $Q^H_e(v_i)$ on the cycle $Q^H(v_i)$.
Similarly, we put $a_{i+1}$ in $Q^H_e(v_{i+1})$ between the only
two vertices of $Q^H_e(v_{i+1})$ on the cycle $Q^H(v_{i+1})$.
The newly created vertices are connected with a pair of new edge edges of $e$
of the same direction as the neighboring two edge edges of $e$ in $\ext(H)$.
The insertion of the new vertices and edges
is propagated to all the ancestors of $H$.
The new vertices subdivide edges between a vertex face and an edge face,
so no boundary of an original face is changed. Properties
\ref{prop:b1}-\ref{prop:b4} remain valid for $\ext(H)$ and its ancestors;
property~\ref{prop:b5} will hold again for the two newly created children.

Let $v_iv_{i+1}$ be a hole part of $C_H$ inside the hole $h$ of $H$.
By Lemma~\ref{lem:cycle-separator}, this hole part joins corners
$\angle_i,\angle_{i+1}$ at distinct vertices of $h$; the corners may be neighboring.
Let $j\in\{i,i+1\}$.
Let $e_1^j,e_2^j$ be the two, possibly equal
edges of $\angle_j$.
If $e_1^j\neq e_2^j$, let $e_1^j$ be the edge following $e_2^j$ at this
occurrence of the corner in the bounding walk of $h$, traversed with $h$
on its right. This convention also applies to triangular holes and
to repeated vertices or edges on a hole boundary.
\begin{itemize}
\item If $e_1^j=e_2^j$, then $\deg_H(v_j)=1$ and for $j=i$ ($j=i+1$, resp.)
we set $b_i$ ($a_{i+1}$, resp.) to be any element of $X^H(v_j)$ -- recall that $X^H(v_j)\neq\emptyset$ in this case.
\item Observe that if $e_1^j\neq e_2^j$, then $e_1^j$ and $e_2^j$ cannot both correspond to edge parts of $C_H$.

    If $e_1^i$ lies on $C_H$, then set $b_i$ to be the first
    edge vertex of $Q^H_{e_2^i}(v_i)$.
    Otherwise, set $b_i$ to be the last edge vertex of $Q^H_{e_1^i}(v_i)$.

    Similarly, If $e_1^{i+1}$ lies on $C_H$, then set $a_{i+1}$ to be the first
    edge vertex of $Q^H_{e_2^{i+1}}(v_{i+1})$.
    Otherwise, set $a_{i+1}$ to be the last edge vertex of $Q^H_{e_1^{i+1}}(v_{i+1})$.
\end{itemize}

The vertices $a_i,b_i$ are defined in such a way that
$a_i\neq b_i$, $a_i,b_i\in Q^H(v_i)$ and $a_i$ and $b_i$
are not next to each other on $Q^H(v_i)$.
To see that, consider the parts $v_{i-1}v_i$, $v_iv_{i+1}$ of $C_H$.
If both parts are edge parts, then $a_i$ and $b_i$ are both edge vertices
and are not border edge vertices.
If both parts are hole parts, then $a_i$ and $b_i$ are both last
edge vertices of distinct edges incident to $v_i$.
If exactly one of $v_{i-1}v_i$ and $v_iv_{i+1}$ (say, $v_{i-1}v_i$) is an edge-part
of $C_H$, then $a_i$ is an edge vertex but not a border edge vertex of $Q^H(v_i)$.
If $\deg_H(v_i)=1$, then $b_i\in X^H(v_i)$ and otherwise $b_i$ is a border
edge vertex of some other edge of $Q^H(v_i)$.

For each $i$, let $P_{a_ib_i}$ ($P_{b_ia_i}$, resp.)
be the clockwise path $a_i\to b_i$ ($b_i\to a_i$, resp.) on $Q^H(v_i)$.
We add edges $e_{a_ib_i}=a_ib_i$ and $e_{b_ia_i}=b_ia_i$ inside the cycle $Q^H(v_i)$
of $\ext(H)$,
so that the interiors of directed cycles
$P_{a_ib_i}e_{b_ia_i}$, $P_{b_ia_i}e_{a_ib_i}$ do not intersect.
Note the added edges do not violate the embedding (see~\ref{prop:b5}).
Similarly as before, these edges are added to all ancestors of $\ext(H)$,
consistently with the embedding.

Now, we define a Jordan curve $C'$ going through vertices $a_1,b_1,a_2,b_2,\ldots,a_k,b_k$:
each part from $a_i$ to $b_i$ goes ``between'' the added edges $e_{a_ib_i}$ and $e_{b_ia_i}$
(in other words, strictly inside the face bounded by these two edges),
whereas each part from $b_i$ to $a_{i+1}$ is either inside an original hole
of $H$ or goes between a pair of newly-added edge edges.
We define $\child_i(\ext(H))$ ($i=1,2$) to be the part of $\ext(H)$ weakly inside (outside, resp.) $C'$
iff $\child_i(H)$ is the subgraph of $H$ weakly inside (outside, resp.) $C_H$.
Consult Figure~\ref{fig:extended-split} for better understanding.

\begin{figure}
  \centering
  \includegraphics{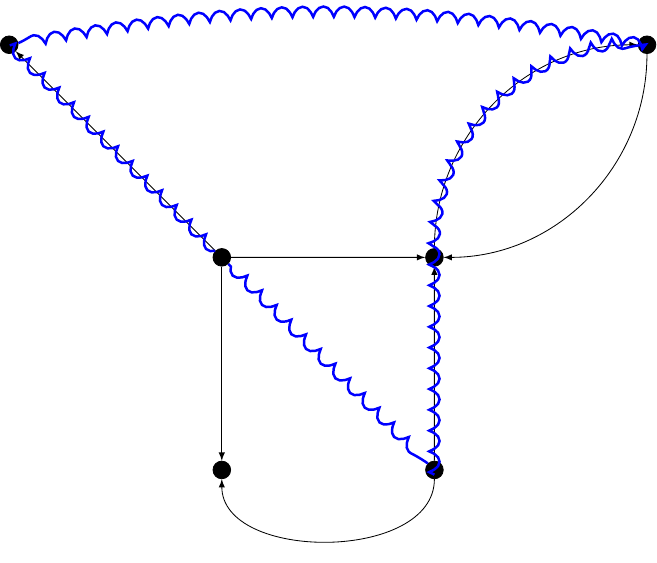}
  \includegraphics{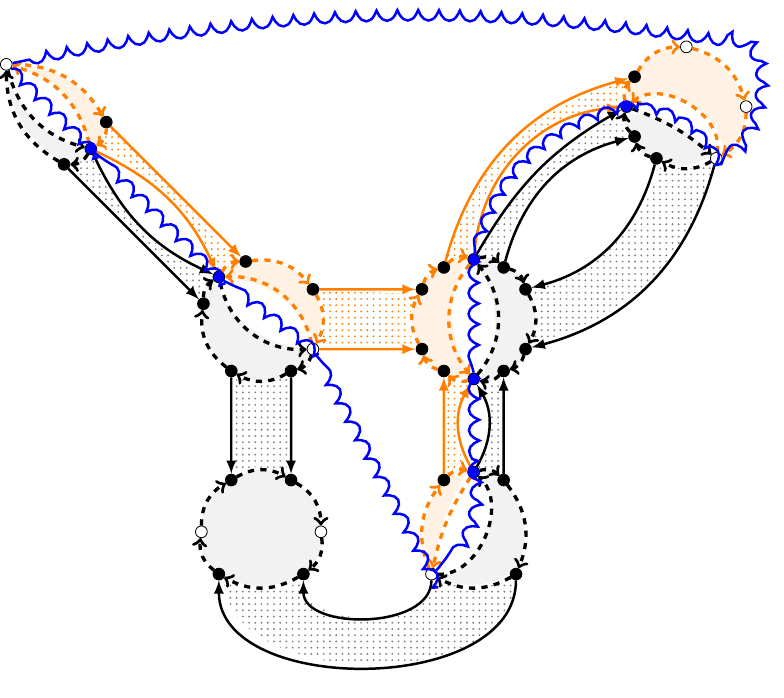}
\caption{The top picture shows a graph $H$ and a cycle separator $C_H$ (blue).
The bottom picture shows how the corresponding cycle separator $C'$ splits
$\ext(H)$ into two parts. The orange part is $\child_1(\ext(H))$, while
the rest -- $\child_2(\ext(H))$.
Compare to Figure~\ref{fig:extended} to see which edges were added.
The blue vertices are the newly-introduced edge vertices.
The white vertices represent extra vertices in either $\child_1(\ext(H))$ or $\child_2(\ext(H))$.}
\label{fig:extended-split}
\end{figure}

It is easy to see that $\child_i(\ext(H))$ satisfies all structural
properties posed on $\ext(\child_i(H))$ at the beginning of this section.
Observe that $\bnd{\child_i(\ext(H))}$ contains the vertices
$\{a_1,b_1,\ldots,a_k,b_k\}$ and those vertices of $\bnd{\ext(H)}$ that
lie strictly inside (outside resp.) $C'$.
The definition of $C'$ also guarantees, that $\child_1(\ext(H))$ and
$\child_2(\ext(H))$ have no edges in common.

Note that there are only two cases when for some $v_i\in C_H$ and some $j\in\{1,2\}$,
$\deg_H(v_i)=\deg_{\child_j(H)}(v_i)$.
The first one is when $v_{i-1}v_i$ and $v_{i}v_{i+1}$ are both edge parts
and additionally the corresponding edges are neighboring in the edge ring of $v_i$.
But then both $a_i$ and $b_i$ are non-border edge vertices of $v_i$ and
thus the partition with $C'$ does not introduce extra vertices of $v_i$ in
either $\child_1(\ext(H))$ or $\child_2(\ext(H))$ that were not extra vertices
in $\ext(H)$.
The second case is when (wlog.) $v_{i-1}v_i$ is an edge part of $C_H$ corresponding to an
edge of the walk bounding $h$ and $v_{i}v_{i+1}$
is a hole part going through the corner $\angle_i$ of $h$ incident to $e$.
In this case $\deg_{\child_{3-j}(H)}(v_i)$ is $1$ and $\child_{3-j}(\ext(H))$ has
one new extra vertex of $v_i$.
However, no new extra vertices of $v_i$ are introduced in $\child_j(\ext(H))$.

Below we sketch how the properties
\ref{prop:b1}-\ref{prop:b5} are inductively satisfied for $\child_i(\ext(H))$
immediately after its creation and that they still hold later in the process.
\begin{itemize}
\item[\ref{prop:b1}]
Recall that the boundary $\bnd{\child_i(\ext(H))}$ is defined, by Lemma~\ref{lem:cycle-separator},
as $V(\child_i(\ext(H)))\cap (C'\cup\bnd{\ext(H)})$
    and thus is contained in $C'\cup\bnd{\ext(H)}$.
The property~\ref{prop:b1} for $\child_i(\ext(H))$ follows from the property~\ref{prop:b1}
for $\ext(H)$ and the fact
that for each $i$ we have $\orig(a_i)=\orig(b_i)=v_i$, i.e.,
$\orig(C'\cap V(\ext(H)))=C_H\cap V(H)$ and
$\bnd{\child_i(H)}=(C_H\cap V(H))\cup(V(\child_i(H))\cap\bnd H)$.
\item[\ref{prop:b2}]
When $\child_i(\ext(H))$ is created, it does
not contain edge edges that are not border edge edges.
Also, when we create the descendants of $\child_i(\ext(H))$,
new edge vertices are inserted strictly inside existing edge fragments,
away from original faces, and the added vertex edges lie inside vertex faces.
Thus no new boundary vertices are introduced in $\child_i(\ext(H))$.
\item[\ref{prop:b3}]
It is sufficient to show it immediately after creating $\child_j(\ext(H))$,
as no new extra vertices are added to $\child_j(\ext(H))$ later on.
The vertices $a_i,b_i$ are defined in such a way that the size of $X^{\child_i(H)}(v_i)$
can be greater than $X^H(v_i)$ by at most $1$ and moreover this can only happen if
$\deg_{\child_j(H)}(v_i)<\deg_{H}(v_i)$.
Thus $|X^{\child_j(H)}(v_i)|-|X^H(v_i)|\leq \deg_{H}(v_i)-\deg_{\child_j(H)}(v_i)$.
    By adding this inequality to $|X^H(v_i)|\leq \deg_G(v_i)-\deg_{H}(v_i)+1$,
which follows
from the property~\ref{prop:b3} for $H$, we get the property~\ref{prop:b3}
for $\child_i(H)$.
\item[\ref{prop:b4}]
  Without loss of generality, assume that $\child_i(\ext(H))$ is the subgraph
of $\ext(H)$ weakly inside $C'$.
Note that each part of $C'$ corresponding to the hole
part of $C_H$ inside a hole $h$, lies inside the original face of $\ext(H)$
    corresponding to $h$, by \ref{prop:b4} for $\ext(H)$.
On the other hand, for all holes $h$ that do not contain hole parts of $C_H$,
$h$ is entirely on one side of $C_H$, and thus the original face of $\ext(H)$
corresponding to $h$ is on one side of $C'$.
By combining this fact with property~\ref{prop:b4} of $\ext(H)$
and Lemma~\ref{lem:cycle-separator},
we get that the vertices $\bnd{\ext(H)}\cap V(\child_i(\ext(H)))$
lie either on the original faces of $\child_i(\ext(H))$ inside $C'$
corresponding to the holes of both $H$ and $\child_i(H)$ inside $C_H$,
or on the ``external'' original face of $\child_i(\ext(H))$, which
corresponds to the ``external'' hole $h^*$ of $\child_i(H)$.
The vertices $\bnd{\child_i(\ext(H))}\setminus\bnd{\ext(H)}$ are clearly
contained in $C'$ and thus lie on the ``external'' original face of $\child_i(\ext(H))$
as well.
The inherited original holes and this new original hole contain omitted parts
of $G'$, whereas the original non-hole faces retain their boundaries.
In the final graphs, all vertex and edge faces of the child are faces of $G'$:
whenever such a face is subdivided, the same insertion is performed in the child.
Thus precisely the stated original faces are holes. Moreover, later insertions
only split vertex and edge faces and leave all original-face boundaries,
and hence the boundaries of already created pieces, unchanged.
\item[\ref{prop:b5}]
The new edges and vertices are introduced in $\ext(H)$ in such a way
that a child is given a separate pair of edge edges for each edge
of $\child_i(H)$.
Moreover, for each $v_i$, $\child_i(\ext(H))$ inherits
a contiguous part of the cycle going through $Q^H(v_i)$ and
a single new edge $a_ib_i$ or $b_ia_i$ that closes the cycle.
For all vertices $v\notin C_H$, no non-border vertex edges of $v$
are added to $\ext(H)$ and thus no such edges can exist in $\child_i(\ext(H))$.
\end{itemize}

It remains to verify the size and running time of the construction.
By Lemma~\ref{lem:klein}, the original tree has $O(n)$ leaves and height
$O(\log n)$, and the expanded tree has the same shape.
Every edge of $G'$ lies in a leaf of the expanded tree. By
Lemma~\ref{lem:border-size}, each such leaf has $O(1)$ edges, so
$|E'|=O(n)$ and $|V'|=O(n)$.
The boundary bounds follow from the same lemma, and property
\ref{prop:b4}, Lemmas~\ref{lem:hole-disjoint}~and~\ref{lem:hole-simple}
provide the required holes. The children cover their parent and have
disjoint edge sets. Hence the obtained decomposition is simple.

At the time a piece $\ext(H)$ is split, properties~\ref{prop:b3}~and~\ref{prop:b5}
bound its current size by $O(|V(H)|+|E(H)|)$.
The sum of these sizes is $O(n\log n)$ by Lemma~\ref{lem:klein}.
Store edges by shared records and maintain the cyclic incidence lists in
each piece. A subdivision or insertion then takes constant time per
affected ancestor. All operations only add edges and vertices or subdivide
edges, so their number is $O(n)$ by the final size bound, and each affects
$O(\log n)$ ancestors. This gives total time $O(n\log n)$, completing the
proof of Lemma~\ref{lem:simple-ext}.